\documentclass[11pt,a4paper]{article}
\usepackage{a4wide}
\usepackage[colorlinks=true,citecolor=blue]{hyperref}

\usepackage{graphicx}
\usepackage{float}
\usepackage{subcaption}
\usepackage{epstopdf}
\usepackage{floatflt} 
\usepackage{footnote}

\usepackage{amsmath}
\usepackage{amsthm}
\usepackage{amssymb}
\usepackage{breqn} 
\usepackage{dsfont} 
\usepackage{algorithmic}
\usepackage{algorithm}
\usepackage{xfrac}
\usepackage{bbm} 

\newtheorem{remark}{Remark}

\usepackage{tikz}
\newlength\fwidth
\usetikzlibrary{arrows}
\usetikzlibrary{decorations.pathreplacing}
\usetikzlibrary{plotmarks}
\usepackage{pgfplots} 
\usepackage{pgfgantt}
\usepackage{pdflscape}
\pgfplotsset{compat=newest} 
\pgfplotsset{plot coordinates/math parser=false}
\usepackage{varwidth}
\usepgfplotslibrary{fillbetween}
\pgfplotsset{compat = 1.3}

\usepackage{booktabs}
\usepackage{multirow}

\usepackage{color}

\newtheorem{proposition}[]{Proposition}[]
\newtheorem{lemma}[]{Lemma}[]

\providecommand{\norm}[1]{\lVert#1\rVert}
\providecommand{\abs}[1]{\lvert#1\rvert}
\providecommand{\rhomax}[0]{\rho_{crit}}
\providecommand{\rhocrit}[0]{\rho_{crit}}
\providecommand{\einhalb}[0]{\frac{1}{2}}
\providecommand{\indicatormaxdens}[0]{\, \mathbbm{1}_{ \left( \rho(x)   > \rhocrit \right)}}
\providecommand{\dt}[0]{\Delta t}
\providecommand{\dx}[0]{\Delta x}

\begin{document}
	
	\title{Density dependent diffusion models for the interaction of particle ensembles with boundaries}
	
	\author{Jennifer Weissen\footnotemark[1], Simone G\"ottlich\footnote{University of Mannheim, Department of Mathematics, 68131 Mannheim, Germany (jennifer.weissen@uni-mannheim.de, goettlich@uni-mannheim.de)}, Dieter Armbruster\footnote{Arizona State University, School of Mathematical and Statistical Sciences, Tempe, AZ 85287-1804, USA (dieter@asu.edu)}}

	\maketitle

	\begin{abstract}

		The transition from a microscopic model for the movement of many
		particles to a macroscopic continuum model for a density flow is studied. The microscopic model for the free flow is completely deterministic, described by an interaction potential that leads to a coherent motion where all particles move in the same direction with the same speed
		known as a flock. Interaction of the flock with boundaries, obstacles
		and other flocks leads to a temporary destruction of the coherent motion that macroscopically can be modeled through density dependent
		diffusion. The resulting macroscopic model is an advection-diffusion
		equation for the particle density whose diffusion coefficient is density
		dependent. Examples describing i) the interaction of material flow on
		a conveyor belt with an obstacle that redirects or restricts the material
		flow and ii) the interaction of flocks (of fish or birds) with boundaries
		and iii) the scattering of two flocks as they bounce off each other are
		discussed. In each case, the advection-diffusion equation is strictly hyperbolic before and after the interaction while the interaction phase is
		described by a parabolic equation. A numerical algorithm to solve the
		advection-diffusion equation through the transition is presented.
		
	\end{abstract}


	
	
	\section{Introduction} \label{sec:intro}

	We study the transition from microscopic models of interacting particles to the  macroscopic limit describing their motion as 
	coherent ensembles. Such problems naturally arise in the description of biological swarms such as flocks of birds~\cite{ArmMarTha2017,ChuDorMar2007,DorChuBer2006}, schools of fish~\cite{Aoki1982}, ant~\cite{BoiCapMor2000} or bacterial colonies~\cite{KocWhi1998}, the movement of pedestrian crowds~\cite{GoeKnaSch2018,HelMol1995} or transport of material ~\cite{GoeHohSch2014,RosWeiGoa2019}. 
	In some production facilities, e.g. bottling plants,  work in process (i.e. bottles) is  transported on conveyor belts from one processing station to the next. Bottles are positioned in an initial configuration on a conveyor belt.   This  initial configuration remains intact 
	and is transported with constant speed as long as individual objects do not collide. On the biological level, social animals gather together and move collectively, often in  synchronized and  coherent patterns~\cite{Selforganization_biologicalsystems}.
	Thus  individuals organize into swarms and build formations that remain stable over time. When such swarms move with constant velocity and form  a  well-defined translational steady state where all individuals head in a common direction they 
	are called a flock. While each individual has their own initial velocity the interactions between individuals lead to a  stable formation where, in the absence of external perturbations, all individuals have exactly the same velocity. In that sense, 
	the movement of a flock acts like a conveyor belt.  In this paper, we model  situations where the stable formation is perturbed or destroyed by external interactions with boundaries and the movement of the particle ensemble is reorganized.
	
	On the microscopic level we consider the motion of  $N$ individuals moving according to Newton's law. On a conveyor belt the  relative motion of single objects is determined by frictional forces which  emerge due to contact  with the conveyor belt, its geometric restrictions and collisions with other objects~\cite{GoeHohSch2014}. The resulting system of ordinary differential equations describes  transport of 
	all parts on the moving belt. For biological swarms, the individual motion is typically governed by a velocity selection mechanism and attractive and repulsive interactions between individuals leading to similar systems of coupled ordinary differential equations. Among the microscopic models for swarming, the Vicsek model~\cite{VicCziBen1995} and the attraction-repulsion model~\cite{DorChuBer2006} have received considerable attention in the literature. 
	
	There have been a few studies at the microscopic level to describe the phenomenology of the interaction of flocks with geometrical boundaries. Studies in finite domains of the Vicsek model, see~\cite{ArmMotTha2017} and references therein, and the attraction-repulsion model~\cite{ArmMarTha2017} illustrate how the geometry of the domain influences steady state flocking solutions. In particular,  boundaries generate internal excitations in the swarm which causes the flocking solution to break apart.  Depending on the geometry of the domain the flock 
	reorganizes with a different direction, similar to particle scattering. In the Vicsek model, the noise level and the influence horizon impact the formation of the collective. In the attraction-repulsion model  the relative strength of attraction and repulsion in comparison to the self-propelling forces determine whether flocks scatter elastically or inelastically from a boundary. Similarly when two or more flocks~\cite{ArmMarTha2017} collide they may  cross almost without interacting or their formation may be temporarily  destroyed and 
	they  re-emerge as two flocks with different directions of motions or they may merge into one flock. 
	
	The prototype for micro-macro transitions is kinetic gas theory developed by Boltzmann in the $19^{th}$ century which linked the macroscopic measurable quantities of heat and temperature to microscopic particle motion. One principal advantage of macroscopic models is that they are computationally independent from the number of individuals. Microscopic models are computationally expensive for large numbers of individuals that often appear in nature where animal societies might contain thousands or millions of individuals. Another reason to use macroscopic models is that macroscopic solutions like flocks can be observed within microscopic simulations and the emergent properties of their motion can be studied and described well at a macroscopic level. 
	
	Hence continuum models are developed for the limit as the number of particles goes to infinity.  
	Being aware of the fact that
	this limit is a much better approximation for gases than it is for birds or bottles we derive macroscopic models for self-organized flow via a mean field limit from the underlying microscopic models. We are especially interested in the description of stable motions  which are disturbed  when interacting with boundaries or obstacles. The boundary interaction leads to internal perturbation which is modeled by diffusion on the macroscopic scale. 
	
	The general type of the  macroscopic equation that we will derive is the (strongly degenerate) advection-diffusion equation for the particle density $\rho = \rho(x,t)$
	\begin{align}
		\partial_t \rho + \nabla_x \cdot (f(\rho,x) - k(\rho) \nabla \rho) = 0 \qquad (x,t) \in \mathbb{R}^2 \times (0,T), \label{eq:advectiondiffusion_unbounded}
	\end{align}
	where diffusion is generally coming from random motion in the underlying microscopic model.
	Under several assumptions, among them the requirements that $k(\rho)$ is sufficiently smooth, Volpert and Hudjaev~\cite{VolHud1969} showed existence of a BV entropy solution for~\eqref{eq:advectiondiffusion_unbounded} in several space dimensions and unbounded domains. Uniqueness of weak solutions in the class of bounded integrable functions for the purely parabolic case and  nondecreasing $k(\rho)$ has been proven by Brezis and Crandall~\cite{BreCra1979}. Yin~\cite{Jin1990} showed uniqueness of weak solutions in $L^\infty \cap BV$ for the Cauchy Problem of~\eqref{eq:advectiondiffusion_unbounded} and strictly increasing $k(\rho)$. Carillo~\cite{Carrillo1999} showed uniqueness of entropy solutions for particular boundary value problems with Dirichlet boundary conditions.
	
	In our case, the hyperbolic part of equation~\eqref{eq:advectiondiffusion_unbounded} allows for the correct modeling of material transport  and flocks. In the undisturbed situation our model is diffusion free, i.e. $k(\rho)=0$. Whenever the bulk of material or the flock is disturbed, diffusion becomes active and is generated from the ergodic properties of the large number of particle interactions.  These disturbances are not constant in time but instead depend on time and space. 
	
	The paper is organized as follows: In Section~\ref{sec:derivationmacrolimit}, we present the microscopic model and give a short overview on existing macroscopic limits before we formally derive our macroscopic equation. Starting from the macroscopic equation for material transport in Section~\ref{sec:application_materialflow}, we generalize the concept to model biological swarms in Section~\ref{sec:application_swarming}. Section~\ref{sec:numericalresults} introduces the operator splitting method to compute approximate solutions to the advection-diffusion equation. We study the properties of the material flow model numerically in Section~\ref{sec:numerics_materialflow} and compare our results to a non-local macroscopic model as well as to experimental data. 
	Section~\ref{sec:numerics_swarming} discusses 
	numerical results for the movement of flocks in bounded domains and  scattering interactions of two swarms. 
	Section~\ref{sec:conclusion} summarizes our results. 
	
	
	\section{Derivation of the macroscopic limit} \label{sec:derivationmacrolimit}
	
	\subsection{Discrete model and macroscopic treatment}
	We consider the deterministic, second order microscopic model describing the movement of particles driven by a velocity selection mechanism and a pairwise interaction force \cite{ChuDorMar2007}.  The equations of motion for 
	$N$ particles  are
	\begin{align} 
		\begin{split}
			\label{eq:generalmicromodel}
			\frac{dx_i}{dt} &= v_i \\
			m \frac{dv_i}{dt} &= G(v_i) + \sum_{j \neq i } F(x_i-x_j),
		\end{split} \qquad i=1, \dots, N
	\end{align}
	where $x_i, v_i \in \mathbb{R}^2$ are the position and velocity of the particle $i$ and $m$ is the mass. $G(v)$ is the velocity selection mechanism and $F(x)$ is the interaction force depending on the positions and distance of particles.
	
	To describe material flow on conveyor belts, the interaction of individual particles in the force term $F$ is reduced to short-range repulsion when two particles collide. First order equations derived from~\eqref{eq:generalmicromodel} have been considered in~\cite{GoeKlaTiw2015} where a non-local second order model is derived from the microscopic model~\eqref{eq:generalmicromodel} via mean field limit and a second order macroscopic model which couples the continuity equation
	\begin{equation} \label{eq:continuityequation}
		\partial_t \rho + \nabla_x \cdot (\rho v) = 0,
	\end{equation}
	with a momentum equation. The continuity equation is closed with a closure velocity derived from the momentum equation. The resulting first order model is the continuity equation with non-local velocity $v(\rho,x,t)$.
	Alternatively, the (non-degenerate) advection-diffusion equation~\eqref{eq:advectiondiffusion_unbounded} is proposed where the density is conserved while particles travel with the constant average speed of the conveyor belt $f(\rho,x) = v_T \rho$ and their movements are subject to diffusion with strength $k(\rho, x,t) = C \rho $ with $C >0$, see~\cite{GoeKlaTiw2015}.
	
	In connection with animal swarming, the focus has been mainly on models in which the density of the population $\rho$ satisfies the advection-diffusion equation~\eqref{eq:advectiondiffusion_unbounded} where $f(\rho,x) = v \rho$. The velocity $v=v(\rho,x,t)$ is a non-local velocity and $k(\rho,x,t)$ is the diffusion coefficient~\cite{MogEde1999,TopBer2004}. The non-locality models spatially decaying social forces including attraction and repulsion between the individuals and is based on the fact that interactions between individuals via sight, sound or smell often take place at a larger distance~\cite{TopBerLew2006}. Diffusion in the continuum limit leads to disordering and dispersal within the swarm. Density independent diffusion leads to disintegration of swarms on large time scales, while density dependent diffusion can stabilize swarms~\cite{MogEde1999}. In comparison to non-locality in the advection term, non-local effects in the diffusion term do not lead to qualitatively new patterns in the movement of animals~\cite{MogEde1999}. 
	
	The usual approach to derive first order continuum equations from first order microscopic models involves a Fokker-Planck approximation~\cite{OkuLev2001}. The seminal derivation of the advection-diffusion-terms from stochastic microscopic models was considered in~\cite{Alt1980,Gru1994,OthDunAlt1988}. The microscopic models describe motion of individual cells or organisms subject to random jumps or turns modeled by (biased) random walks~\cite{Alt1980, OthDunAlt1988}. In addition, cells sense and are influenced by  the number of neighbours which  is assumed to be distributed with Poisson probability~\cite{Gru1994}. The main assumption is that individual movements in the microscopic model include sufficiently large random motion. Then, individual-based stochastic simulations agree well with the behaviour described by the limit equation. 
	
	Microscopic second order models for swarming are connected to their macroscopic counterpart using kinetic theory as a middle stage. The kinetic equation for the single particle probability distribution function is derived from the particle scale and then related to the macroscopic limit equation with additional assumptions~\cite{DorPanCar2008,AceBosCar2019,ChuDorMar2007,DegMot2008}. The book~\cite{NalParTos2010} (Part III, Section 2-4) contains a comprehensive summary on microscopic swarming models and their second order continuum limits. The momentum equations are non-local equations which depict interactions over a broader range of space. For a general overview on microscopic and macroscopic models for swarming, we especially refer the reader to the papers~\cite{ChuDorMar2007,TopBerLew2006} and references therein. 
	
	To our knowledge, first order models derived from second order microscopic models do not exist in the literature in the context of swarming.  
	Thus, in the next section, we follow the approach of~\cite{GoeKlaTiw2015} to motivate a first order macroscopic limit from the second order macroscopic model~\eqref{eq:generalmicromodel}. 
	Specifically, the   full macroscopic equations which couple the continuity equation ~\eqref{eq:continuityequation} to  a non-local momentum equation are reduced to a local first order limit equation through the identification of  a local closure velocity.
	\subsection{Mean field limit and macroscopic equation}
	Starting from the microscopic model~\eqref{eq:generalmicromodel}, we derive a macroscopic limit equation via mean field limit as middle stage. Initially, we consider $N_0$ ball shaped particles with radius $R_0$ and mass $m_0$. We let $N \rightarrow \infty, R \rightarrow 0$, where $N$ is the number of individuals and $R$ is the radius. We rescale and keep the total mass and the total surface covered by the particles constant
	\begin{equation*}
		N m = N_0 m_0,  \qquad  N \pi R^2 = N_0 \pi R_0^2. 
	\end{equation*} 
	
	Let $f^{(N)}(x,v,t) = \frac{m_0 N_0}{N} \sum_{i=1}^{N} \delta_{x_i(t) \times v_i(t)}$, then 
	
	\begin{equation*}
		\int \rho(x,t) dx = \int \int f^{(N)}(x,v,t)\; dv dx = m_0 N_0.
	\end{equation*}

	The corresponding mean field equation is 
	\begin{align*}
		\partial_t f &+ v \cdot \nabla_x f + S_f = 0, \\
		S_f &= \nabla_v \cdot \left( \frac{1}{m}  \left(G(v) + \int \int F(x-y) f(y,w,t) dw dy \right) f(x,v,t) \right), 
	\end{align*}
	where $F(0) = 0$. Using a mono-kinetic closure, we arrive at the macroscopic limit equation
	\begin{align}
		\partial_t \rho + \nabla \cdot (\rho u) &= 0  \label{eq:macro2ndorder_1}\\
		G(u) + \int F(x-y) \rho(y) dy &= 0, \label{eq:macro2ndorder_2}
	\end{align}
	where we have left out the time and space dependency of $\rho,u$ whenever the meaning is clear, see~\cite{GoeKlaTiw2015} for further details.
	Apparently, the equations~\eqref{eq:macro2ndorder_1} and~\eqref{eq:macro2ndorder_2} are coupled via the velocity $u$.
	To derive a closed model consisting of a single equation only for the density in~\eqref{eq:macro2ndorder_1}, an explicit closure relation for the velocity $u$ is needed. However, the velocity is only implicitly given by $G(u)$ in~\eqref{eq:macro2ndorder_2}. So the key idea in the following section is to determine the velocity $u$ depending on the force term $F$, i.e. $u=G^{-1}(\int F(x-y)\rho(y)dy)$. As we will see, the choice of the interaction potential leads to different types of advection-diffusion equations for~\eqref{eq:macro2ndorder_1}.

	\subsubsection{Macroscopic limit for interaction potentials with compact support}
	
	We are especially interested in interaction forces of the form 
	\begin{align} \label{eq:interactionforce}
		F(x) = H(d_R - \norm{x}) F_R(x),
	\end{align}
	where $H$  is the Heaviside function and $F_R$ is the gradient field of a potential $U_R$, i.e., $F_R = \nabla U_R$. Note that $F_R$ is odd, i.e. $F_R(-x) = -F_R(x)$. The interaction force~\eqref{eq:interactionforce} between two particles is only active up to the distance $d_R$. We assume that the distance $d_R$ can be expressed depending on the radius of the particle, i.e., a constant ratio $R/d_R$. The distance $d_R$ is the horizon up  to which a single particle can sense others. 
	
	Consider the force $F(x_i -x_j)$ acting on particle $i$ induced by the particle $j$.  On the microscopic level, particles are described by their center of mass. If their centers of mass are in close proximity  they have a repulsive impact on each other such that the particles experience a force  which pushes them apart. Usually particle $i$ is pushed in the direction  $-(x_i-x_j)$ opposite to particle $j$. On the macroscopic level, the density distribution describes the spatial concentration of the mass. Thus there are no particles  as a density distribution represents the collection of infinitely many infinitely small particles with zero distance and hence  the distance between particles is not defined. 
	
	Since macroscopically the integral over the density is the mass, repulsion should only be active, if  microscopic particles overlap, indicated by a cumulated density that  is too high. Therefore, we switch from scaling the force in terms of the distance microscopically to scaling it macroscopically in terms of the mass. 
	
	Without any changes, we rewrite the microscopic force term~\eqref{eq:interactionforce} as follows
	\begin{align*}
		F(x_i-x_j)&=  H(d_R -\norm{x_i-x_j}) F_R(x_i-x_j) \, \mathbbm{1}_{ \left(H(d_R - \norm{x_i-x_j}) = 1 \right) } \\
		&=H(d_R -\norm{x_i-x_j}) F_R(x_i-x_j) \, \mathbbm{1}_{ \left( \sum_j m H(d_R - \norm{x_i-x_j}) > m \right)}.
	\end{align*}
	
	We interpret the expressions $m$ and $\sum_j m H(d_R - \norm{x_i-x_j})$ in the additional indicator function macroscopically and use a Taylor expansion to reformulate 
	\begin{align*}
		m &\sim \frac{m_0N_0}{N} =\frac{m_0 R^2}{R_0^2}, \\
		\sum_j m H(d_R - \norm{x_i-x_j}) &\sim \int_{B_{d_R(x)}} \rho(y) dy \approx \int_{B_{d_R(x)}} \rho(x) +\nabla \rho(x) \cdot (x-y) dy \\  
		&= \pi (d_R)^2 \rho(x) + \nabla \rho(x) \cdot \underbrace{ \int_{B_{d_R}(0)} \tilde{z} \, d \tilde{z}}_{=0}. 
	\end{align*}
	This way, we derive an expression for the density threshold $\rhocrit$ above which diffusion is observed
	\begin{align}
		\left( \sum_j m H(d_R - \norm{x_i-x_j}) > m \right) &\sim \rho(x) > \frac{m_0 R^2}{\pi (d_R)^2 R_0^2 } =: \rhocrit,  \label{eq:derivation_rhomax}
	\end{align}
	which is meaningful even for $R \rightarrow 0$, as the ratio $R/d_R$ is fixed by assumption. Then, we plug in the expression for the force term into equation~\eqref{eq:macro2ndorder_2}, exploit that $F_R$ is odd and use again Taylor expansion to obtain
	
	\begin{align*}
		\int F(x-y) \rho(y) dy &= \int_{B_{d_R(x)}} F_R(x-y) \rho(y) \, \mathbbm{1}_{ \left( \int_{B_{d_R(x)}} \rho(y) dy  > \frac{m_0 N_0}{N} \right)} \, dy \\
		& =\int_{B_{d_R(x)}} F_R(x-y) \rho(y) \,\indicatormaxdens \, dy \\
		&= \int_{B_{d_R(0)}} F_R(-z) \rho(x+z) \, dz \, H(\rho(x) - \rhocrit) \\
		&\approx - \int_{B_{d_R(0)}} F_R(z) \langle \nabla  \rho(x), z \rangle \, dz  \, H(\rho(x) - \rhocrit).
	\end{align*}

	We are interested in the limit of the force term for $R \rightarrow 0$. Therefore, we assume that the interaction force is chosen such that the limit can be reformulated as follows
	\begin{align}
		\begin{split}
			&\lim_{R \rightarrow 0} \int_{B_{d_R(0)}}  F_R(z) \langle \nabla \rho(x), z \rangle \, dz = \overline{C}  \nabla \rho(x)  , \text{ where }  \overline{C}< \infty. \label{eq:requirementforceterm_limitedsupport}
		\end{split}
	\end{align}
	Then, the local approximation of the force is denoted by
	\begin{align*}
		\Psi_{d_R}(\rho, \nabla \rho) &:= \lim_{R \rightarrow 0} \int_{B_{d_R(0)}} F_R(z) \langle \nabla  \rho(x), z \rangle  \, dz  \, H(\rho(x) - \rhomax) \\
		&= \overline{C} \nabla \rho(x)  \, H(\rho(x) - \rhomax).
	\end{align*}
	If the self-propelling force $G$ is invertible, we can solve~\eqref{eq:macro2ndorder_2} for the velocity $u = G^{-1}(\Psi_{d_R}( \rho, \nabla \rho))$ and substitute $u$ in~\eqref{eq:macro2ndorder_1}, such that we obtain the first order macroscopic limit equation
	\begin{align} \label{eq:general_limitequation}
		\partial_t \rho + \nabla_x \cdot \left(\rho G^{-1}(\Psi_{d_R}( \rho, \nabla \rho ))\right) = 0.
	\end{align}

	\begin{remark}
		In the case of interaction potentials with unlimited support we achieve a similar result.
		Let us consider interaction forces $F(x) = F_\infty(x)$ with support $\text{supp}(F) = \mathbb{R}^2$ and instead of fixing the ratio $R/d_R$, we assume $d_R = \infty$ in equation~\eqref{eq:interactionforce}. The resulting critical density where diffusion becomes active is $\rhocrit =0$. Defining
		\begin{align}\label{eq:requirementforceterm_unlimitedsupport}
			\Psi_\infty(\rho, \nabla \rho) ~& = \lim_{R \rightarrow 0} \int_{\mathbb{R}^2} F_{\infty}(z) \langle \nabla  \rho(x), z \rangle  \, dz  = \bar{C} \nabla \rho(x) < \infty,
		\end{align}
		leads to the limit equation 
		\begin{align} \label{eq:general_limitequation_infinteinteraction}
			\partial_t \rho + \nabla_x \cdot \left(\rho G^{-1}(\Psi_\infty( \rho, \nabla \rho ))\right) = 0.
		\end{align}
	\end{remark}
	
	In the following section, we present some exemplary particle systems and interaction forces to illustrate the two types of limit equations, i.e. the degenerate advection-diffusion equation~\eqref{eq:general_limitequation} or the non-degenerate advection-diffusion equation~\eqref{eq:general_limitequation_infinteinteraction}.

	
	\subsection{Applications} \label{sec:applications}
	\subsubsection{Material Flow} \label{sec:application_materialflow}
	
	The microscopic model for material flow on a conveyor belt $\Omega \subset \mathbb{R}^2$ describes the transport of  identical and homogeneous parts with mass $m_0$ and radius $R_0$ with velocity $v_T \in \mathbb{R}^2$, see~\cite{GoeHohSch2014}. 
	The regularized bottom friction 
	\begin{align} \label{eq:bottomfriction}
		G(v)  = - \gamma_b (v - v_T),
	\end{align}
	corrects deviations of the parts' velocities from the conveyor belt velocity where $\gamma_b$ is the bottom viscous damping.
	
	The interaction force $F$ is given by a spring-damper model of the form
	\begin{align} \label{eq:interactionforce_materialflow}
		F(x) = H(2R_0- \norm{x}) F_{2R_0}(x).
	\end{align}
	
	\begin{proposition} \label{prop:continuumlimit_materialflow}
		For the material flow model~\eqref{eq:generalmicromodel} with bottom friction~\eqref{eq:bottomfriction} and interaction force~\eqref{eq:interactionforce_materialflow} obeying~\eqref{eq:requirementforceterm_limitedsupport}, the macroscopic limit given by equation~\eqref{eq:general_limitequation} is the degenerate advection-diffusion equation
		\begin{align} 
			\begin{split}
				\label{eq:continuumlimit_materialflow_csupport}
				\partial_t \rho + \nabla_x \cdot \left(\rho  v_T -  k(\rho) \nabla \rho \right) = 0, 
			\end{split}
		\end{align}
		with threshold density $ \rhocrit = \frac{m_0 }{\pi 4 R_0^2 }$ and $k(\rho) = \frac{\bar{C}}{\gamma_b} \rho H(\rho- \rhomax)$.
	\end{proposition}
	\begin{proof}
		The interaction force $F$ has to satisfy~\eqref{eq:requirementforceterm_limitedsupport}. Without loss of generality, we carry out the analysis for $F_{R}(z) = k_m \frac{z}{\norm{z}} \frac{(2R - \norm{z})}{R^4}, k_m >0$. It holds that
		\begin{align*}
			&\lim_{R \rightarrow 0} \int_{B_{2_R(0)}}  F_R(z) \langle \nabla \rho(x), z \rangle \, dz \\
			&= \lim_{R \rightarrow 0} \int_{B_{2R}(0)} k_m \frac{z}{\norm{z}}\frac{(2R-\norm{z})}{R^4} (\partial_{x^{(1)}} \rho z^{(1)} + \partial_{x^{(2)}} \rho z^{(2)}) \, dz \\ 
			&= \lim_{R \rightarrow 0}  k_m \pi \nabla \rho  \int_0^{2R} r^2 \frac{(2R - r)}{R^4} dr  =\frac{8}{3} k_m \pi \nabla \rho =\bar{C} \nabla \rho,
		\end{align*}
		where we set the macroscopic diffusion constant $\bar{C}$ equal to the microscopic term $\frac{8}{3} k_m \pi$. 
		Thus the scale of the interaction force $k_m$ determines the  strength of the diffusion coefficient $\bar{C}$. We obtain
		\begin{align*}
			\Psi_{2R}(\rho, \nabla \rho) = \bar{C} \nabla \rho H(\rho - \rhocrit).
		\end{align*}
		Using $G^{-1}(y) = v_T - \frac{y}{\gamma_b}$, the velocity $u$ is then given by
		\begin{align*}
			u = v_T - \frac{ \bar{C} \nabla \rho H(\rho - \rhocrit)}{\gamma_b}.
		\end{align*}
	\end{proof}
	The Kirchhoff transformation of $k(\rho)$ given by $
	b(\rho) = \int_{0}^{\rho} k(y) dy, 
	\nabla b = k(\rho) \nabla \rho $
	allows to recast~\eqref{eq:continuumlimit_materialflow_csupport}
	\begin{align}
		\partial_t \rho + \nabla \cdot (f(\rho,x) - \nabla b(\rho)) &= 0 \qquad \qquad x \in \Omega, \label{eq:advectiondiffusion_linear} 
	\end{align}
	which we consider in the following in a bounded domain $\Omega \subset \mathbb{R}^2$ equipped with boundary and initial conditions
	\begin{align}
		(f(\rho,x) - \nabla b(\rho)) \cdot \vec{n} &= 0 \qquad \qquad x \in \partial \Omega \label{eq:boundarycondition}, \\
		\rho(x,0) = \rho_0(x), \label{eq:initialcondition}
	\end{align}
	where $\vec{n}$ is the outer normal vector at the boundary $\partial \Omega$. The limit equation is the degenerate advection-diffusion equation with convective flux $f(\rho,x) = \rho v_T$ and diffusive flux $\nabla b = k(\rho) \nabla \rho $. Equation~\eqref{eq:continuumlimit_materialflow_csupport} is formally parabolic, but is purely hyperbolic when $k(\rho) \geq 0$ vanishes. Since $k'(\rho) \geq 0$ is zero on a set of positive measures, the equation is strongly degenerate. In the purely hyperbolic case, i.e. $k(\rho) = 0$ the equation reduces to the transport equation describing transport at the velocity $v_T$ of the conveyor belt
	$$ \partial_t \rho + \nabla_x \cdot (\rho v_T) = 0.$$
	If the conveyor belt is stopped (i.e. $v_T = 0$), we obtain the purely parabolic case
	$$\partial_t \rho  = \Delta b(\rho),$$
	where material spreads out from regions which are densely packed, i.e. regions with density values above the density value $\rhocrit$. 
	
	The density $\rhocrit$ matches the microscopic situation. The  density in~\eqref{eq:continuumlimit_materialflow_csupport} corresponds to the concentration of mass of a single part with total mass $m_0$ in the circle with radius $R_0$, i.e, a maximum density $\rhocrit$ at which parts are packed as closely as possible without overlaps. This definition of the maximum density is derived from ball shaped particles. 
	With initial data $\norm{\rho^0}_{L^\infty} < \rhocrit$ in unbounded domains, the equation describes a trivial transport process without diffusive influence.
	
	The most interesting situation occurs when the  compactly supported interaction force leads to the distinction between  $\rho < \rhocrit$ and $\rho \geq \rhocrit$.  In particular the mixed hyperbolic-parabolic equation~\eqref{eq:continuumlimit_materialflow_csupport} is purely hyperbolic for $\rho < \rhocrit$ and contains a  parabolic region  if $\rho > \rhocrit$. 
	
	The diffusion in equation~\eqref{eq:continuumlimit_materialflow_csupport} is only activated when the transport of the material is disturbed and the material density increases above the maximum density. The diffusion models dispersal, is spatially local and acts in the opposite direction of the gradient. Once diffusion is active, it is linear in both $\rho$ and $\nabla \rho$. The ratio between  $v_T$ and $\frac{\overline{C}}{\gamma_b}$ determines whether the problem is advection or diffusion dominated. As the  viscous damping $\gamma_b$ increases, diffusion decreases and the transport process becomes advection dominated.  
	
	A typical experiment involves the placement of bottles on a conveyor belt. Without obstacles, the conveyor belt moves them with a constant velocity and no relative velocity between them hence the initial configuration is preserved. The spatial translation $\rho(x,t) = \rho^0(x-v_T t)$ is accurately represented by the hyperbolic part of~\eqref{eq:continuumlimit_materialflow_csupport} and sharp material gradients are maintained. As soon as an obstacle
	is placed on the belt, the transportation process is disturbed. As a consequence,  the initial configuration of the bottles is destroyed. They are pressed closer together and when the maximum density in the macroscopic model is exceeded, as a result of the diffusion, the velocity of the bottles will deviate from the transport velocity $v_T$. 
	
	Generally in advection-diffusion equations the Péclet number
	measures the relative strength of advection and diffusion.
	The Péclet number vanishes for pure diffusion and is infinite in case of pure advection. 
	We remark that the Péclet number for the degenerate advection-diffusion equation~\eqref{eq:continuumlimit_materialflow_csupport} is defined locally and strongly varies with the density. For regions where the undisturbed transport process takes place, the Péclet number is infinite, while for regions with high density values and strong diffusion i.e. near disturbances,  very small Péclet numbers can be observed. 
	
	The following Lemma shows that for interaction potentials with unlimited support and interaction potentials chosen such that~\eqref{eq:requirementforceterm_unlimitedsupport} holds true, the limit equation is given by the advection-diffusion equation where the diffusion term is always active for $\rho >0$.

	\begin{lemma}\label{lem:continuumlimit_materialflow_unlimitedsupport}
		For the material flow model~\eqref{eq:generalmicromodel} with bottom friction~\eqref{eq:bottomfriction} and interaction force~\eqref{eq:interactionforce} with $d_R =\infty$ obeying~\eqref{eq:requirementforceterm_unlimitedsupport}, the macroscopic limit given by equation~\eqref{eq:general_limitequation_infinteinteraction} is the (non-degenerate) advection-diffusion equation
		\begin{align} 
			\begin{split}
				\label{eq:continuumlimit_materialflow_unlimitedsupport}
				\partial_t \rho + \nabla_x \cdot \left(\rho  v_T -  k(\rho) \nabla \rho \right) = 0, 
			\end{split}
		\end{align}
		where $k(\rho) = C \rho, C = \frac{\bar{C}}{\gamma_b}$.
	\end{lemma}
	
	\begin{proof}
		For the exemplary interaction potential $F_\infty(x) = k_m\frac{x}{\norm{x}} \exp(-3\norm{x}^3)$, we obtain 
		\begin{align*}
			\Psi_\infty(\rho, \nabla \rho) = \underbrace{\frac{1}{9} k_m \pi}_{\overline{C}} \nabla \rho H(\rho).
		\end{align*}
		The rest of the proof is analogous to the proof of Proposition~\ref{prop:continuumlimit_materialflow}.
	\end{proof}
	
	Qualitatively, this limit equation is also derived in~\cite{GoeKlaTiw2015}. Again, the diffusion is density dependent and scales linearly with the density. The existence and uniqueness results, compare Section~\ref{sec:intro}, are applicable for~\eqref{eq:continuumlimit_materialflow_unlimitedsupport} in an unbounded domain, but not for~\eqref{eq:continuumlimit_materialflow_csupport}, since $k(\rho)$ is not continuous and does not meet the regularity requirements to obtain mentioned results. Moreover, existence and uniqueness of the solution is also not proven for~\eqref{eq:continuumlimit_materialflow_csupport},~\eqref{eq:continuumlimit_materialflow_unlimitedsupport} on a bounded domain $\Omega \subset \mathbb{R}^2$ equipped with the boundary conditions~\eqref{eq:boundarycondition}. However, we see in Section~\ref{sec:numericalresults} that  numerical examples provide good results in bounded domains. 
	
	\subsubsection{Pedestrian Crowds} \label{sec:application_peds}
	The dynamical behaviour of pedestrian crowds has been modeled by  the deterministic microscopic social force model~\cite{GoeKnaSch2018,HelMol1995}. Our  microscopic model~\eqref{eq:generalmicromodel} can be 
	written in this way by considering  the velocity selection mechanism $G(v)$ as the destination force $G^{dest}(x,v) = \frac{1}{\tau} (v^C D(x) -v)$. Here, $v^C >0$ is a comfort speed which is achieved within the relaxation time $\tau >0$ and $D:\mathbb{R}^2 \rightarrow \mathbb{R}^2$ describes the direction to a  destination. The interaction force $F$ in this context describes attraction and repulsion between pedestrians. We assume that $F$ is chosen such that~\eqref{eq:requirementforceterm_limitedsupport} is satisfied. Exploiting the inverse function $G^{-1}(y) = v^C D(x) - y \tau$, the limit equation is the degenerate advection-diffusion equation with space dependent advection term
	\begin{align}
		\partial_t \rho + \nabla_x \cdot \left(\rho \left(v^c D(x) - \tau\bar{C} \nabla \rho H(\rho -\rhomax) \right)\right) = 0.
		\label{eq:continuumlimit_pedestrians}
	\end{align}

	\subsubsection{Swarming} \label{sec:application_swarming}
	
	We consider the microscopic attraction-repulsion model for biological swarms
	\begin{align}
		\begin{split}
			\label{eq:attraction_repulsion}
			\frac{dx_i}{dt} &= v_i \\
			m \frac{dv_i}{dt} &=  (\alpha - \beta \norm{v}^2)v + \lambda \nabla_{x_i} \sum_{j \neq i } U(x_i-x_j),
		\end{split}
	\end{align}
	where $U$ is a potential. The potential strength $\lambda>0$ scales the impact of the potential forces relative to the self-propelling force~\cite{DorPanCar2008}. The velocity selection mechanism is given by
	\begin{align}\label{eq:selfpropellingforce_swarms}
		G(v) = (\alpha - \beta \norm{v}^2)v,
	\end{align}
	with $\alpha , \beta \geq 0 $. For velocity independent forces, i.e., $\alpha, \beta = 0$,  equations~\eqref{eq:attraction_repulsion} form a Hamiltonian system with energy conservation. Acceleration  $\alpha v$ and deceleration dynamics $-\beta \norm{v}^2 v$ are added to the system with $\alpha, \beta >0$. In the case $\alpha, \beta \geq 0$, after a transition phase, the Hamiltonian system is recovered when particles travel at equilibrium speed $\norm{v} = \sqrt{\alpha / \beta}$. An important property to describe the large particle limit of the system is $H$-stability of the potential $U$. In particular, $H$-stability of the potential ensures that particles do not collapse for $N \rightarrow \infty$. For non-$H$-stable (catastrophic) systems, increasing particle numbers $N$ reduce the particle spacing and particles collapse. In either case, for $H$-stable as well as catastrophic systems, stationary solutions can emerge and are characterized by $\alpha, \beta$ and the potential characteristic, see~\cite{DorChuBer2006}.
	
	For $\alpha,\beta > 0$ and a Morse potential where the potential minimum exists, stationary states of the large particle systems are observed in which the particles form a coherent structure and travel at speed with $\norm{v}= \sqrt{\alpha / \beta}$ for the $H$-stable as well as the catastrophic situation~\cite{DorChuBer2006}. For $N \rightarrow \infty$, a well-defined spacing is maintained for the $H$-stable potential, which is also true for finite N  for the catastrophic potential.

	To model the collective behavior macroscopically for $\alpha, \beta >0$\footnote{We remark that for $\alpha , \beta >0$, $G$ is not invertible and we cannot directly apply~\eqref{eq:general_limitequation}.}, we consider~\eqref{eq:continuumlimit_pedestrians} with average velocity $\overline{v}, \norm{\overline{v}} = v^c = \sqrt{\alpha / \beta}$. The macroscopic limit is then given by
	\begin{align} \label{eq:macroscopicswarming_singleswarm}
		\partial_t \rho + \nabla_x \cdot \left(\rho  \overline{v} - \rho C \nabla \rho H(\rho - \rhocrit) \right),
	\end{align}
	where $\overline{v}$ as the average speed of the swarm. The density threshold $\rhocrit$ above which diffusion in a swarm occurs is not to be understood as a maximum density, when individuals are as closely together as possible, but instead as the preferred density of the undisturbed flock above which diffusion is active. Collisions of the swarm with obstacles lead to internal diffusion of the swarm and can change the direction $\overline{v}$. The collisional behavior will be validated by numerical comparisons of~\eqref{eq:macroscopicswarming_singleswarm} to corresponding microscopic behavior in Section~\ref{sec:numericalresults}.
	
	For $\alpha = 0$ and $\beta > 0$, the macroscopic equation~\eqref{eq:general_limitequation} is given by the following Lemma.

	\begin{lemma} \label{prop:continuumlimit_swarming_beta0}
		For the attraction-repulsion model~\eqref{eq:attraction_repulsion} with velocity selection mechanism~\eqref{eq:selfpropellingforce_swarms}, $\alpha =0, \beta > 0$ and interaction force $F = \nabla U$  obeying~\eqref{eq:requirementforceterm_limitedsupport}, the macroscopic limit given by equation~\eqref{eq:general_limitequation} is the diffusion equation
		\begin{align} \label{eq:diffusion_swarming}
			\begin{split}
				\partial_t \rho -  \nabla_x \cdot \left(  k(\rho) \nabla \rho \right) = 0, 
			\end{split}
		\end{align}
		where 
		\begin{equation*}
			k(\rho) = \rho \frac{\bar{C} \nabla \rho H(\rho- \rhocrit)}{\sqrt[3]{\norm{\bar{C}\nabla \rho H(\rho- \rhocrit) } \beta}},
		\end{equation*} and threshold density $ \rhocrit = \frac{m_0}{\pi  R_0^2 } z^2 $  if $R/d_R = z > 0$.
		For an interaction force~\eqref{eq:interactionforce} with $d_R =\infty$ obeying~\eqref{eq:requirementforceterm_unlimitedsupport}, the macroscopic limit given by equation~\eqref{eq:general_limitequation_infinteinteraction} is the diffusion equation~\eqref{eq:diffusion_swarming} with $\rhocrit = 0$.
	\end{lemma}

	\begin{proof}
		We restrict ourselves to the case with $R/d_R = z  > 0$. Since the force term satisfies~\eqref{eq:requirementforceterm_limitedsupport}, we have $\Psi_{d_R}(\rho, \nabla \rho) = \bar{C} \nabla \rho H(\rho - \rhocrit)$. 
		
		We have to determine $G^{-1}(y)$ for $y \in \mathbb{R}^2$.
		Note that we are able to uniquely determine $G^{-1}(0)$ because $G(v) = - \beta \norm{v} v = 0$ has only one solution.  Therefore, we have
		\begin{align*} 
			G^{-1}(y) = \begin{cases}
				\frac{y}{- \sqrt[3]{\norm{y} \beta}} &\text{if } \norm{y} \neq 0, \\
				0 &\text{ if } \norm{y} = 0, \\
			\end{cases}
		\end{align*}
		and the limit is the diffusion equation
		\begin{align*} 
			\partial_t \rho + \nabla_x \cdot \left(- \rho \frac{\bar{C} \nabla \rho H(\rho- \rhocrit)}{\sqrt[3]{\norm{\bar{C}\nabla \rho H(\rho- \rhocrit)}\beta}  }\right) = 0.
		\end{align*}
	\end{proof}

	
	\section{Numerical evaluation in bounded domains} \label{sec:numericalresults}
	For the numerical evaluation of the different types of (advection-)diffusion equations, we present a tailored discretization scheme with operator splitting. In our numerical examples, we particularly focus on the influence of boundaries on the dynamics of the macroscopic models. 
	
	\subsection{Operator splitting method}\label{sec:boundaryconditions}
	
	We discretize a rectangular spatial domain $(\Omega \cup \partial \Omega) \subset \mathbb{R}^2$  with grid points $x_{ij} = (i \Delta x^{(1)},j \Delta x^{(2)}),$ $ (i,j) \in A = \lbrace 1, \dots N_{x^{(1)}}  \rbrace \times  \in \lbrace 1, \dots N_{x^{(2)}} \rbrace$. The boundary is described by the set of indices $B \subset A$. The time discretization is given by $t^s = s \Delta t$. We compute the approximate solution 
	\begin{equation*}
		\rho(x,t) = \rho_{ij}^s ~ \text{for } \begin{cases}
			x \in C_{ij} \\
			t \in [t^{s}, t^{s+1}),
		\end{cases}
	\end{equation*}
	to~\eqref{eq:advectiondiffusion_linear}-\eqref{eq:initialcondition} where $ C_{ij} = \left[(i-\einhalb) \Delta x^{(1)}, (i+\einhalb) \Delta x^{(1)}\right) \times \left[(j-\einhalb) \Delta x^{(2)}, (j+ \einhalb) \Delta x^{(2)}\right)$.  We use an operator splitting method to separate the advective and diffusive terms.

	Since the hyperbolic part of the material flow limit~\eqref{eq:continuumlimit_materialflow_csupport} and the swarming limit~\eqref{eq:macroscopicswarming_singleswarm} reduces to linear transport, we solve the advective part with flux $f^{(l)}(\rho,x) = v^{(l)}(x) \rho, l=1,2$ using the Upwind scheme combined with  dimensional splitting 
	\begin{align}
		\begin{split} \label{eq:Upwind}
			\tilde{\rho}_{ij}^{s+1} &= \rho_{ij}^s - \frac{\Delta t}{\Delta x^{(1)}} \left(F^{(1)}_{up}(\rho_{ij}^s,\rho_{i+1j}^s) - F^{(1)}_{up}(\rho_{i-1j}^s, \rho_{ij}^s) \right) \\
			\overline{\rho}_{ij}^{s+1} &= \tilde{\rho}_{ij}^s - \frac{\Delta t}{\Delta x^{(2)}} \left(F^{(2)}_{up}(\tilde{\rho}_{ij}^s,\tilde{\rho}_{ij+1}^s) - F^{(2)}_{up}(\tilde{\rho}_{ij-1}^s, \tilde{\rho}_{ij}^s) \right), 
		\end{split}
	\end{align}
	where
	\begin{equation*}
		F^{(1)}_{up}(\rho_{ij}^s, \rho_{i+1j}^s) = \begin{cases}
			\rho_{ij}^s v^{(1)}_{ij}  &\text{ if } v^{(1)}_{ij}\geq 0 \\
			\rho_{i+1j}^s v^{(1)}_{ij}  &\text{ otherwise, }
		\end{cases} 
	\end{equation*}
	with $ v^{(1)}_{ij} =v^{(1)}(x_{ij})$ and $F_{up}^{(2)}$ is defined analogously. The diffusive part is solved with the implicit finite difference method
	\begin{align}
		\rho_{ij}^{s+1} &= \overline{\rho}_{ij}^s + \frac{\dt}{\dx^{(1)} \dx^{(2)}} \left(b_{i-1j}^{s+1} + b_{i+1j}^{s+1} - 4b_{ij}^{s+1} + b_{ij-1}^{s+1} + b_{ij+1}^{s+1} \right), \label{eq:implicitscheme}
	\end{align}
	where $b_{ij}^{s+1} = b(\rho_{ij}^{s+1})$. This combination allows to use the time step restriction given by the CFL condition of the hyperbolic part 
	\begin{align} \label{eq:CFL_operatorsplitting}
		\Delta t \leq \min_{(i,j)} \frac{1}{\frac{\abs{v^{(1)}_{ij}}}{\Delta x^{(1)}} + \frac{\abs{v^{(2)}_{ij}}}{\Delta x^{(2)}}},
	\end{align}
	without any additional restriction from the diffusive part, see also~\cite{HolKarLie2000_partII}. As a result, the splitting method enables for large time steps even for relatively high diffusion. 
	
	Numerically, we approximate the Heaviside function $H(\rho - \rhocrit)$, $k(\rho)=C \rho H(\rho-\rhocrit)$ and $b(\rho) = \int_{0}^\rho k(y) \, dy$ with smooth approximations
	\begin{align}
		\begin{split} \label{eq:Heaviside_xi}
			H_{\xi, \rhocrit}(\rho) &= \int_{\rhocrit}^{\rho} \frac{2}{\xi} \max \bigg \lbrace 0, 1- \bigg \vert \frac{2(y- \rhocrit)}{\xi} - 1 \bigg \vert \bigg \rbrace  \, dy  \\
			k_{\xi,\rhocrit}(\rho )  &= C \rho H_{\xi, \rhocrit}(\rho)  \\
			b_{\xi,\rhocrit}(\rho) &= \int_0^\rho Cy H_{\xi, \rhocrit}(y) \, dy, 
		\end{split}
	\end{align}
	for which
	\begin{align*}
		k_{\xi, \rhocrit}(\rhocrit) = 0, \qquad  &k_{\xi, \rhocrit} (\rhocrit + \xi) = C (\rhocrit + \xi), \\
		k_{\xi, \rhocrit}'(\rhomax) = 0, \qquad  &k_{\xi, \rhocrit}'(\rhomax + \xi) = C.
	\end{align*}
	Figure~\ref{img:numapprox_H}-\ref{img:numapprox_k} show the approximations for $\rhocrit = 1$ and varying diffusion constants $C$. In the following, we choose consistently $\xi =10^{-2}$ in all simulations.

	\begin{figure}[htp]
		\begin{subfigure}{0.45\textwidth}
%
%
\begin{tikzpicture}
\setlength\fwidth{0.9\textwidth}
\begin{axis}[%
width=0.951\fwidth,
height=0.75\fwidth,
at={(0\fwidth,0\fwidth)},
scale only axis,
xmin=0,
xmax=2,
xtick={0,1,1.2},
xticklabels={{0},{1},{1$+ \xi$}},
xlabel style={font=\color{white!15!black}},
xlabel={$\rho$},
ymin=0,
ymax=1.2,
ylabel style={font=\color{white!15!black}},
ylabel={$\text{H}_{\xi,1}$},
axis background/.style={fill=white},
legend style={legend cell align=left, align=left, draw=white!15!black}
]
\addplot [color=black]
  table[row sep=crcr]{%
0	0\\
0.98989898989899	0\\
1.01010101010101	0.00510151930131419\\
1.03030303030303	0.0459137148351059\\
1.05050505050505	0.127538007973798\\
1.07070707070707	0.249974564538777\\
1.09090909090909	0.41322312382638\\
1.11111111111111	0.604938675059589\\
1.13131313131313	0.764105642884292\\
1.15151515151515	0.882460971504373\\
1.17171717171717	0.960004082183683\\
1.19191919191919	0.996734807071596\\
1.21212121212121	1.0000001188745\\
2	0.999999836839773\\
};

\end{axis}

\begin{axis}[%
width=1.227\fwidth,
height=0.92\fwidth,
at={(-0.16\fwidth,-0.101\fwidth)},
scale only axis,
xmin=0,
xmax=1,
ymin=0,
ymax=1,
axis line style={draw=none},
ticks=none,
axis x line*=bottom,
axis y line*=left,
legend style={legend cell align=left, align=left, draw=white!15!black}
]
\end{axis}
\end{tikzpicture}%
			\subcaption{Heaviside function $H_{\xi,1}$ for $\rhocrit = 1$.}
			\label{img:numapprox_H}
		\end{subfigure}
		\hfill
		\begin{subfigure}{0.45\textwidth}
%
%
\begin{tikzpicture}
\setlength\fwidth{0.9\textwidth}
\begin{axis}[%
width=0.951\fwidth,
height=0.75\fwidth,
at={(0\fwidth,0\fwidth)},
scale only axis,
xmin=0,
xmax=2,
xtick={0,1,1.2},
xticklabels={{0},{1},{ 1$+ \xi$}},
xlabel={$\rho$},
ymin=0,
ymax=6,
ylabel={$\text{k}_{\xi,1}$},
axis background/.style={fill=white},
legend style={at={(0.03,0.97)}, anchor=north west, legend cell align=left, align=left, draw=white!15!black}
]
\addplot [color=black]
  table[row sep=crcr]{%
0	0\\
0.98989898989899	0\\
1.01010101010101	0.00515304979930731\\
1.03030303030303	0.0473050395270787\\
1.05050505050505	0.133979321507828\\
1.07070707070707	0.26764953374859\\
1.09090909090909	0.450788862356051\\
1.11111111111111	0.672154083399543\\
1.13131313131313	0.864442747505462\\
1.15151515151515	1.01616717930807\\
1.17171717171717	1.12485326801321\\
1.19191919191919	1.18802734580251\\
1.21212121212121	1.21212135621151\\
2	1.99999967367955\\
};
\addlegendentry{$C=1$}

\addplot [color=black, dashdotted]
  table[row sep=crcr]{%
0	0\\
0.98989898989899	0\\
1.01010101010101	0.0103060995986146\\
1.03030303030303	0.0946100790541577\\
1.05050505050505	0.267958643015657\\
1.07070707070707	0.53529906749718\\
1.09090909090909	0.901577724712102\\
1.11111111111111	1.34430816679909\\
1.13131313131313	1.72888549501092\\
1.15151515151515	2.03233435861613\\
1.17171717171717	2.24970653602641\\
1.19191919191919	2.37605469160502\\
1.21212121212121	2.42424271242303\\
2	3.99999934735909\\
};
\addlegendentry{$C=2$}

\addplot [color=black, dotted]
  table[row sep=crcr]{%
0	0\\
0.98989898989899	0\\
1.01010101010101	0.0154591493979215\\
1.03030303030303	0.141915118581236\\
1.05050505050505	0.401937964523485\\
1.07070707070707	0.80294860124577\\
1.09090909090909	1.35236658706815\\
1.11111111111111	2.01646225019863\\
1.13131313131313	2.59332824251639\\
1.15151515151515	3.0485015379242\\
1.17171717171717	3.37455980403961\\
1.19191919191919	3.56408203740752\\
1.21212121212121	3.63636406863454\\
2	5.99999902103864\\
};
\addlegendentry{$C=3$}

\end{axis}

\begin{axis}[%
width=1.227\fwidth,
height=0.92\fwidth,
at={(-0.16\fwidth,-0.101\fwidth)},
scale only axis,
xmin=0,
xmax=1,
ymin=0,
ymax=1,
axis line style={draw=none},
ticks=none,
axis x line*=bottom,
axis y line*=left,
legend style={legend cell align=left, align=left, draw=white!15!black}
]
\end{axis}
\end{tikzpicture}%
			\subcaption{Strength of diffusion $k_{\xi, 1}$ $\rhocrit = 1$.}
			\label{img:numapprox_k}
		\end{subfigure}
		\caption{Numerical approximations of the Heaviside function~\eqref{eq:Heaviside_xi}.}
	\end{figure}
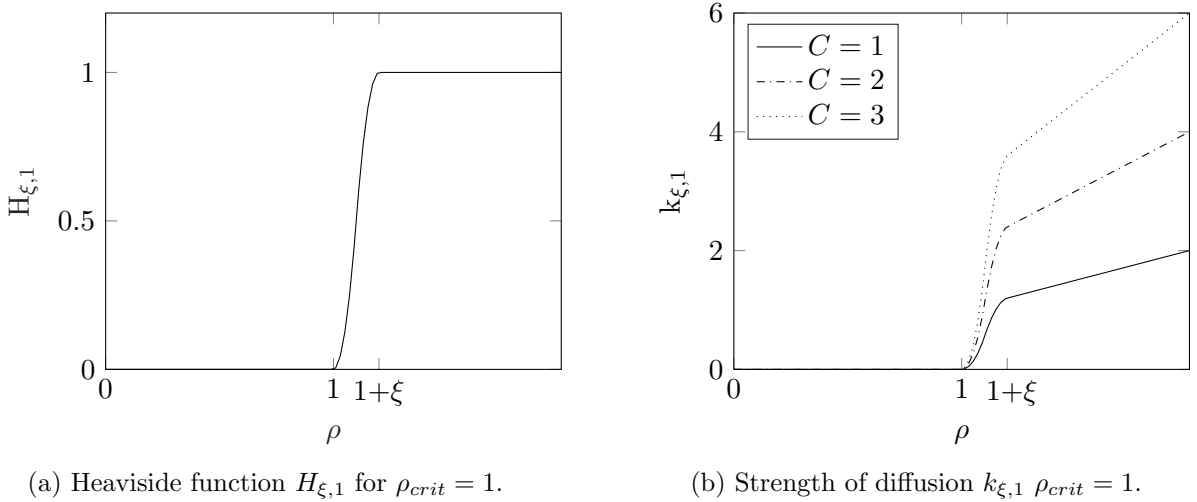

	To satisfy the boundary condition~\eqref{eq:boundarycondition}, we apply zero flux boundary conditions to the advective and to the diffusive flux
	\begin{align*}f \cdot \vec{n} &= 0, x \in \partial \Omega \\
		( k(\rho) \nabla \rho) \cdot \vec{n} &= 0, x \in \partial \Omega,
	\end{align*}
	where $\vec{n}=(n^{(1)},n^{(2)})^T$ is the outer normal vector at the boundary.

	
	\subsection{Numerical results for the material flow model} \label{sec:numerics_materialflow}

	\subsubsection{Comparisons of local models: degenerate vs. non-degenerate} 
	
	To compare the limit equations~\eqref{eq:continuumlimit_materialflow_csupport} and~\eqref{eq:continuumlimit_materialflow_unlimitedsupport}, we perform the following experiment: A bulk of material with sharp edges is placed with uniform spacing on a conveyor belt. The initial density inside the bulk is set to $\rho_0= 0.8$. We send the bulk of material with the conveyor belt velocity $v_T = (1,0)^T$ against a boundary 
	which blocks the transportation, as shown in Figure~\ref{img:MAW_initialconfiguration}. 
	\begin{figure}[h!]
		\centering
		\includegraphics[width=0.45\textwidth]{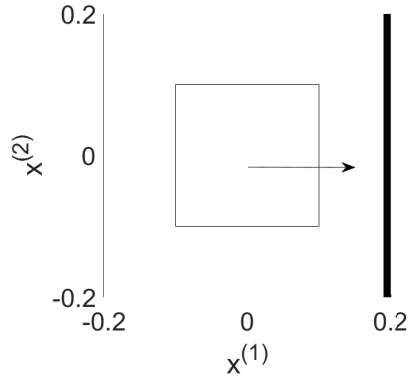}
		\caption{Experimental setup.}
		\label{img:MAW_initialconfiguration}
	\end{figure}
	
	We expect that the initial configuration of the bulk governs a linear transport until it reaches the boundary at time $t=0.1$. 
	Since the density is below the maximum density $\rhomax = 1$,  the material might be compressed up to the density $\rhomax$ when interacting with the boundary.  We use the operator splitting method with step sizes $\dx^{(1)} = \dx^{(2)} = 10^{-2}$ and the CFL condition~\eqref{eq:CFL_operatorsplitting}.
	
	For the degenerate advection-diffusion equation~\eqref{eq:continuumlimit_materialflow_csupport} with diffusion coefficients $k_1=10 \rho H_{\xi, 1}(\rho)$ and $k_2 = \rho H_{\xi, 1}(\rho)$, the bulk is transported until it reaches the boundary at $t=0.1$. The maximum material density $\max_x \rho(x,t)=\rho_0$ is constant for $t < 0.1$ and increases when the bulk hits the boundary. For $k_2(\rho)$, the numerical density exceeds the maximum density, while for $k_1(\rho)$ 
	the numerical density stays at the level of the critical density $\rhomax = 1$ suggesting a maximum principle (Fig. ~\ref{img:max_density_plot}).

	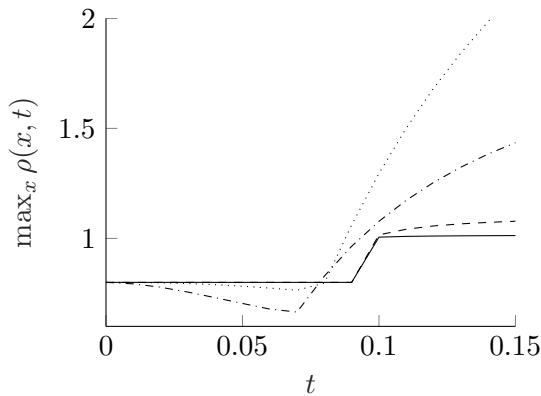
\begin{figure}[htp]
		\begin{center}
%
%
\begin{tikzpicture}
\setlength\fwidth{0.35\textwidth}
\begin{axis}[%
scaled x ticks = false,
width=0.989\fwidth,
height=0.75\fwidth,
at={(0\fwidth,0\fwidth)},
scale only axis,
xmin=0,
xmax=0.15,
xtick={0, 0.05, 0.1, 0.15},
xticklabels={0, 0.05, 0.1, 0.15},
xlabel style={font=\color{white!15!black}},
xlabel={$t$},
ymin=0.6,
ymax=2,
ylabel={$\max_x \rho(x,t)$},
ylabel style={font=\color{white!15!black}},
axis background/.style={fill=white},
axis x line*=bottom,
axis y line*=left,
legend style={at={(0.03,0.97)}, anchor=north west, legend cell align=left, align=left, draw=white!15!black}
]
\addplot [color=black]
  table[row sep=crcr]{%
0	0.8\\
0.0900000000000001	0.8\\
0.1	1.00556832162566\\
0.11	1.00899359892064\\
0.12	1.01063163865188\\
0.13	1.01156348123998\\
0.14	1.01221081126298\\
0.15	1.01272058488643\\
};

\addplot [color=black, dashed]
  table[row sep=crcr]{%
0	0.8\\
0.0900000000000001	0.8\\
0.1	1.01597670937807\\
0.11	1.04148705730209\\
0.12	1.0565512768922\\
0.13	1.06598482893682\\
0.14	1.07268986674118\\
0.15	1.07793475654198\\
};

\addplot [color=black, dotted]
  table[row sep=crcr]{%
0	0.8\\
0.00999999999999979	0.79956765700191\\
0.02	0.798004775712834\\
0.0299999999999998	0.794748137425014\\
0.04	0.789574012894263\\
0.0499999999999998	0.782573790323108\\
0.0600000000000001	0.774031643175097\\
0.0699999999999998	0.764313142145699\\
0.0800000000000001	0.796758657899374\\
0.0899999999999999	1.06111284113282\\
0.1	1.29511216123908\\
0.11	1.50102160537112\\
0.12	1.68251060194268\\
0.13	1.84306118818119\\
0.14	1.98558499220628\\
0.15	2.1124072049775\\
};

\addplot [color=black, dashdotted]
  table[row sep=crcr]{%
0	0.8\\
0.01	0.793324964563467\\
0.02	0.7778373871959\\
0.03	0.755810712903783\\
0.04	0.730441476807267\\
0.05	0.70417339373626\\
0.0600000000000001	0.678544506821724\\
0.0700000000000001	0.665075926652516\\
0.0800000000000001	0.828414634312665\\
0.0900000000000001	0.964143933460738\\
0.1	1.07751347861321\\
0.11	1.17303894369188\\
0.12	1.25421358594782\\
0.13	1.32367569839673\\
0.14	1.38342381615867\\
0.15	1.43499195964985\\
};

\end{axis}

\begin{axis}[%
width=1.287\fwidth,
height=0.965\fwidth,
at={(-0.176\fwidth,-0.143\fwidth)},
scale only axis,
xmin=0,
xmax=1,
ymin=0,
ymax=1,
axis line style={draw=none},
ticks=none,
axis x line*=bottom,
axis y line*=left,
legend style={legend cell align=left, align=left, draw=white!15!black}
]
\end{axis}
\end{tikzpicture}%
		\end{center}
		\caption{ Maximum density as a function of time for Eq.\eqref{eq:continuumlimit_materialflow_csupport} with diffusion coefficient $k_1$ (solid line) and  $k_2$ (dashed line), and  for Eq.\eqref{eq:continuumlimit_materialflow_unlimitedsupport} with diffusion coefficient $k_3$ (dashed dotted) and $k_4$ (dotted line). }
		\label{img:max_density_plot}
	\end{figure}
	
	For the non-degenerate advection-diffusion equation~\eqref{eq:continuumlimit_materialflow_unlimitedsupport}, small diffusion coefficients have to be considered to portray the free flow properly. With the diffusion coefficient $k_3(\rho) = 0.05\rho$, the influence of the diffusion is too high in the free flow phase where the density of the bulk is reduced before the interaction with the boundary happens. For $k_4(\rho) = 0.02 \rho$, the density formation is better captured, but diffusion is not strong enough to portray the boundary interaction correctly. 
	In both cases, when the boundary is reached, diffusion is not strong enough to properly capture the dynamics. Moreover, the impact of the diffusion in the free flow phase misrepresents the time at which the material reaches the boundary. 
	
	Figure~\ref{img:MAW_densityt0k15} shows the density plots of the solutions at $t=0.15$.  Figures~\ref{img:MAW_densityt0k15_1}-\ref{img:MAW_densityt0k15_2} show sharp edges at the rear edge of the bulk and the spread out of material at the boundary. Diffusion smears out the edges of the bulk in Figures~\ref{img:MAW_densityt0k15_3}-\ref{img:MAW_densityt0k15_4} and the material is compressed at the boundary.
	
	Summarizing, Figs. \ref{img:max_density_plot} and \ref{img:MAW_densityt0k15} show that the model~\eqref{eq:continuumlimit_materialflow_csupport} is 
	a better approach, compared to~\eqref{eq:continuumlimit_materialflow_unlimitedsupport}, to describe the evolution of material flow on a conveyor belt.

	\begin{figure}[htp]
		\begin{subfigure}{0.45\textwidth}
			\includegraphics[width=\textwidth]{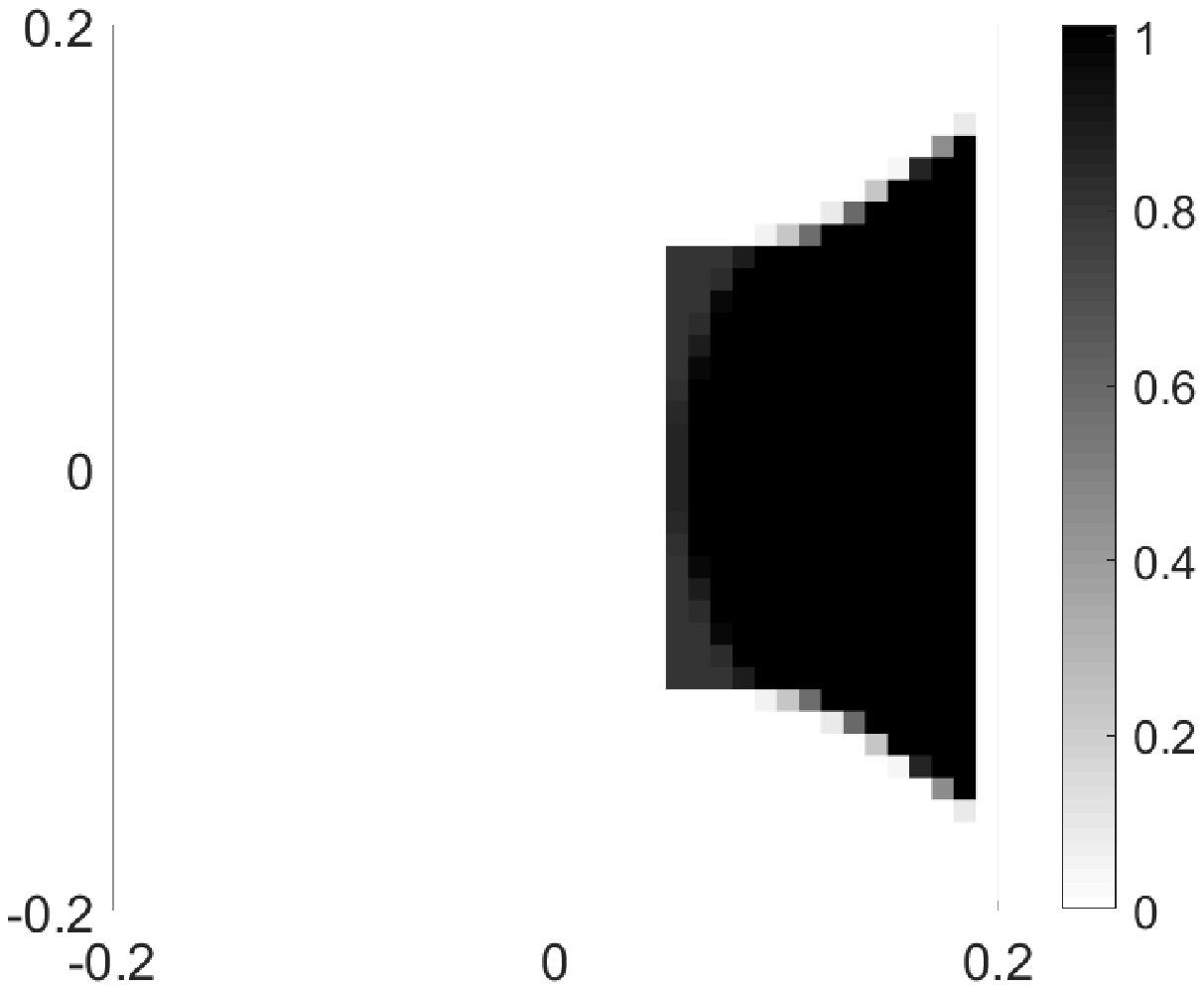}
			\caption{$k_1(\rho).$}
			\label{img:MAW_densityt0k15_1}
		\end{subfigure}
		\hfill
		\begin{subfigure}{0.45\textwidth}
			\includegraphics[width=\textwidth]{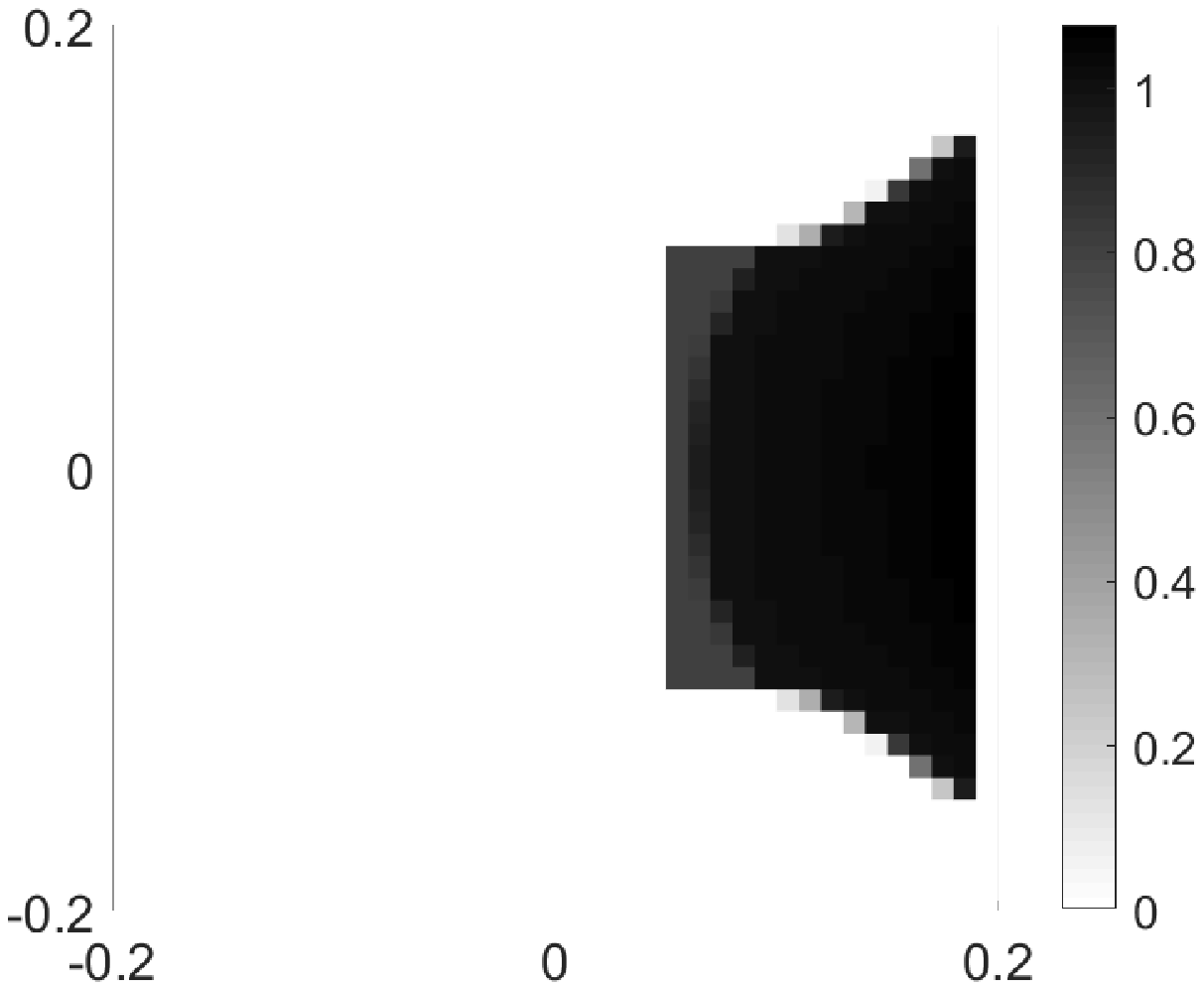}
			\caption{$k_2(\rho).$}
			\label{img:MAW_densityt0k15_2}
		\end{subfigure}
		\begin{subfigure}{0.45\textwidth}
			\includegraphics[width=\textwidth]{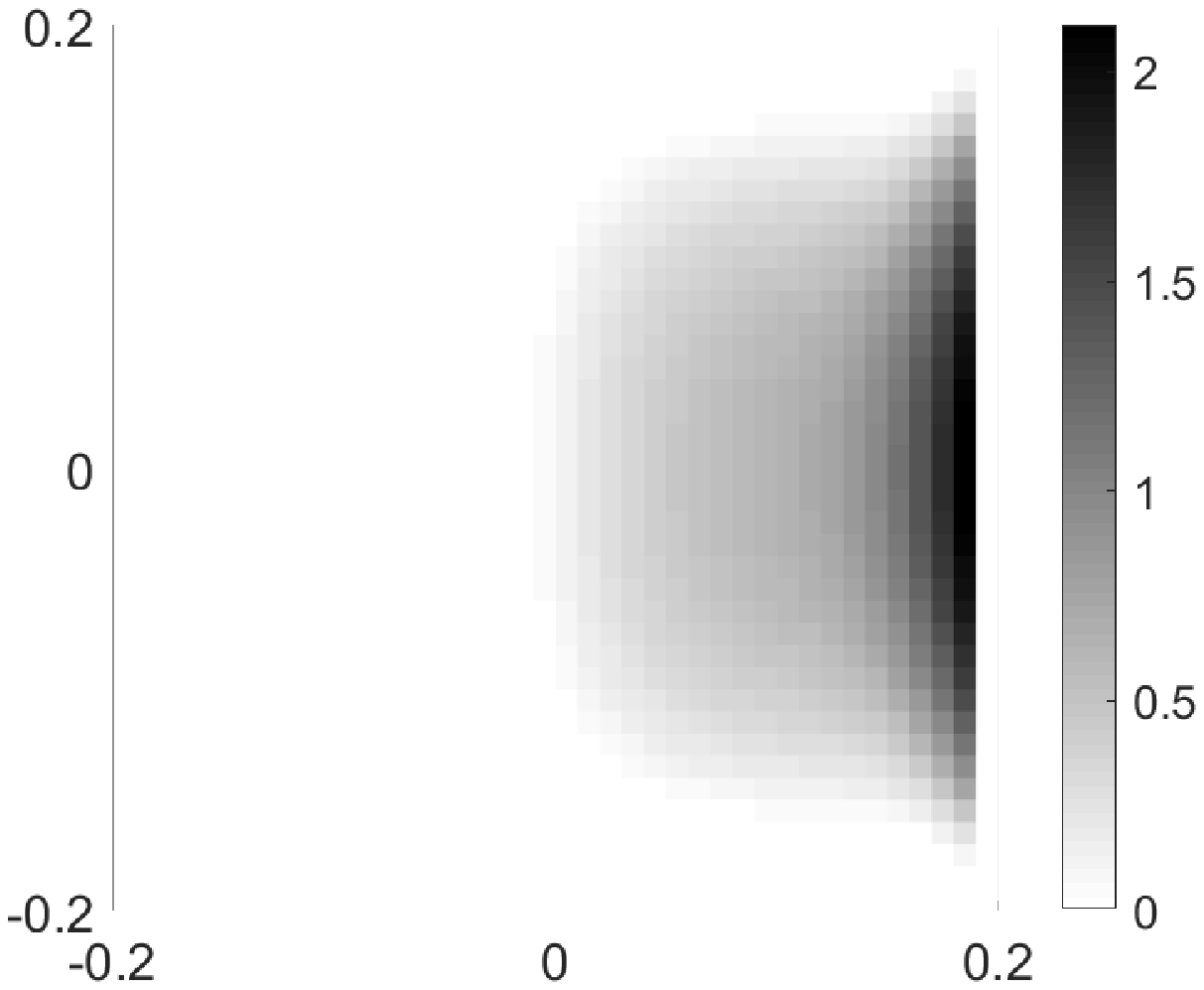}
			\caption{$k_3(\rho).$}
			\label{img:MAW_densityt0k15_3}
		\end{subfigure}
		\hfill
		\begin{subfigure}{0.45\textwidth}
			\includegraphics[width=\textwidth]{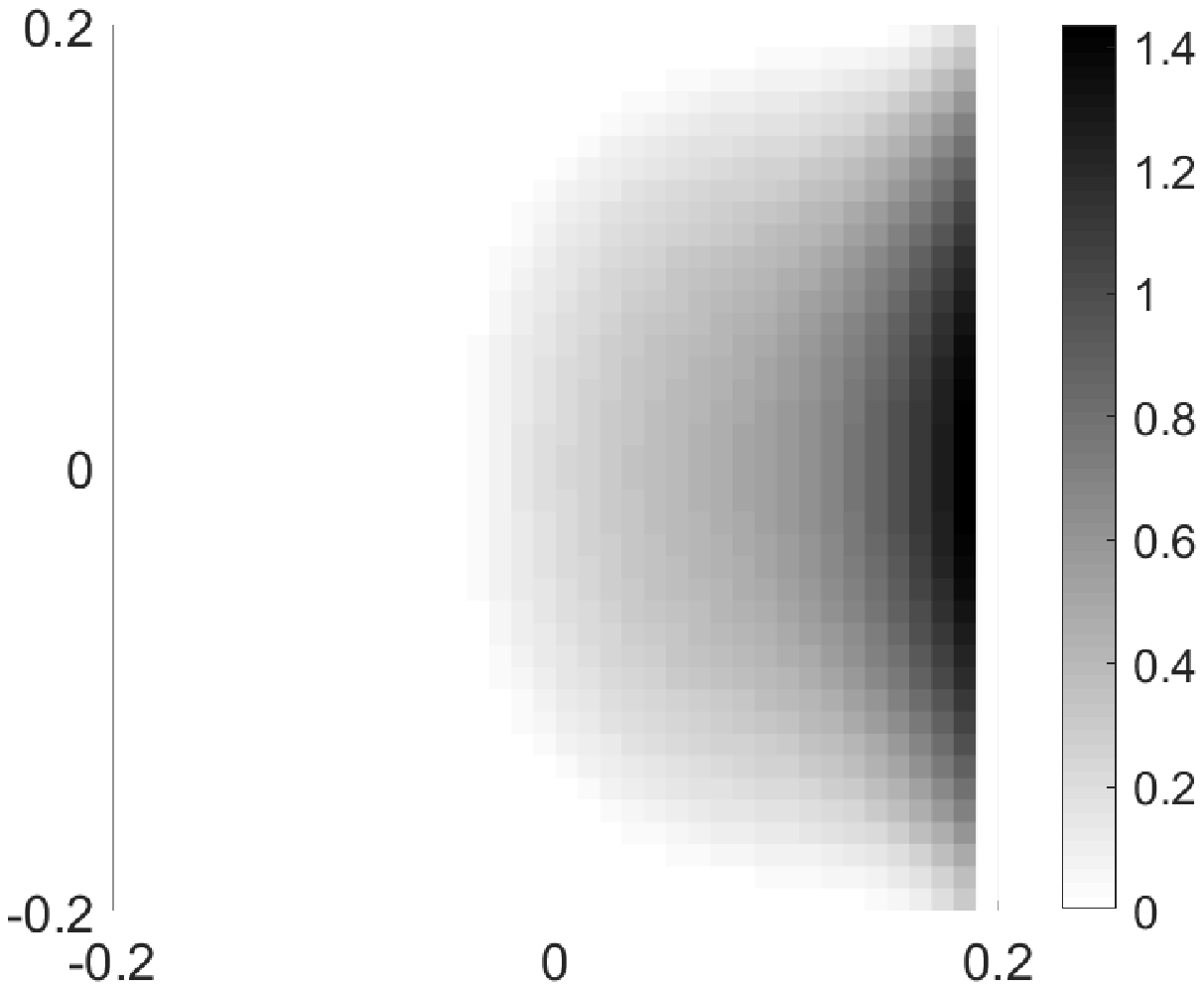}
			\caption{$k_4(\rho).$}
			\label{img:MAW_densityt0k15_4}
		\end{subfigure}
		\caption{Density plots of the solutions at $t=0.15$ for different diffusion coefficients $k(\rho)$.}
		\label{img:MAW_densityt0k15}
	\end{figure}

	\subsubsection{Comparison to a non-local model and experimental data}\label{sec:non-local}
	
	In contrast to the local models we have presented so far,
	non-local models have received an increasing interest 
	over the past decades to mimic transport phenomena in bounded domains such as crowd motion~\cite{ColRos2019} or material flow~\cite{GoeHohSch2014}. Therefore, we recall a non-local version of the degenerate equation~\eqref{eq:continuumlimit_materialflow_csupport}, see~\cite{GoeHohSch2014}. We consider a kernel function $\eta$ and define a non-local model for the density $\rho$ in the following way:
	
	\begin{align}
		\partial_t \rho + \nabla_x \cdot \left (\rho \left( v(x) - \frac{\overline{C}  \, \left( \eta * \nabla \rho \right) \, H(\rho - \rhomax)}{\gamma_b}\right)\right) &= 0. \label{eq:macrolimit_materialflow_totallynonlocal}
	\end{align}
	
	The asterisk $*$ denotes the spatial convolution. At an arbitrary point $x$ in space, we have to consider the non-local gradient $\nabla \left( \eta * \rho \right) (x)$. If we use a mollifier $\eta$ for the non-local gradient, define $\epsilon = \frac{\overline{C} }{\gamma_b}$ and normalize the non-local gradient, we obtain the non-local macroscopic model~\cite{GoeHohSch2014} for material flow
	\begin{align}
		\partial_t \rho + \nabla_x \cdot \left(\rho \left( v(x) -  \epsilon \frac{\nabla \left(\eta * \rho \right)}{\sqrt{1 + \norm{\nabla \left(\eta * \rho \right)}_2^2}}  \, H(\rho - \rhomax) \right )\right) &=0, \label{eq:materialflownonlocal}
	\end{align}
	as a special case of equation~\eqref{eq:macrolimit_materialflow_totallynonlocal}.
	The non-local model~\eqref{eq:materialflownonlocal} has been proven to accurately describe the dynamics in good agreement with experimental data~\cite{GoeHohSch2014,PriKoeGoe2019,RosWeiGoa2019}. For a comparison of the (local) advection-diffusion equation~\eqref{eq:continuumlimit_materialflow_csupport} to the non-local model~\eqref{eq:materialflownonlocal}, we study the experiment of material flow on a conveyor belt $\Omega \in \mathbb{R}^2$, cf.~\cite{GoeHohSch2014,RosWeiGoa2019}. The experiment consists of $N=192$ cylindrical parts which are transported on a conveyor belt and redirected by a deflector with angle $\nu$, see Figure~\ref{img_experimentaldata}. The conveyor belt is modeled with the time-independent velocity field $v(x)$ representing the transport along the belt and the deflector, see \cite{GoeHohSch2014} for a full description. Details for the simulation of \eqref{eq:materialflownonlocal} are discussed in the Appendix \ref{app:non-local}. 
	
	The different modeling approaches are compared in Figure~\ref{img:densitycomparison_realdata}. Obviously, a good agreement of the approximate density for the advection-diffusion equation with both, the experimental data and the solution to the non-local model, is achieved. Figure~\ref{img:maxdens_realdata} depicts the maximum value of the density $\max_{x \in \Omega} \rho(x,t)$. The critical density $\rhomax = 1$ is achieved with the advection-diffusion equation ($H_{\xi,1}$) when congestion at the deflector occurs, while the non-local equation and the advection-diffusion equation with $H_{tan}$ smear out the density profile and the critical density is not reached at any time.  
	
	\begin{figure}[htb]
		\centering
		\begin{minipage}[c]{0.45\textwidth}
			\includegraphics[width=1\textwidth]{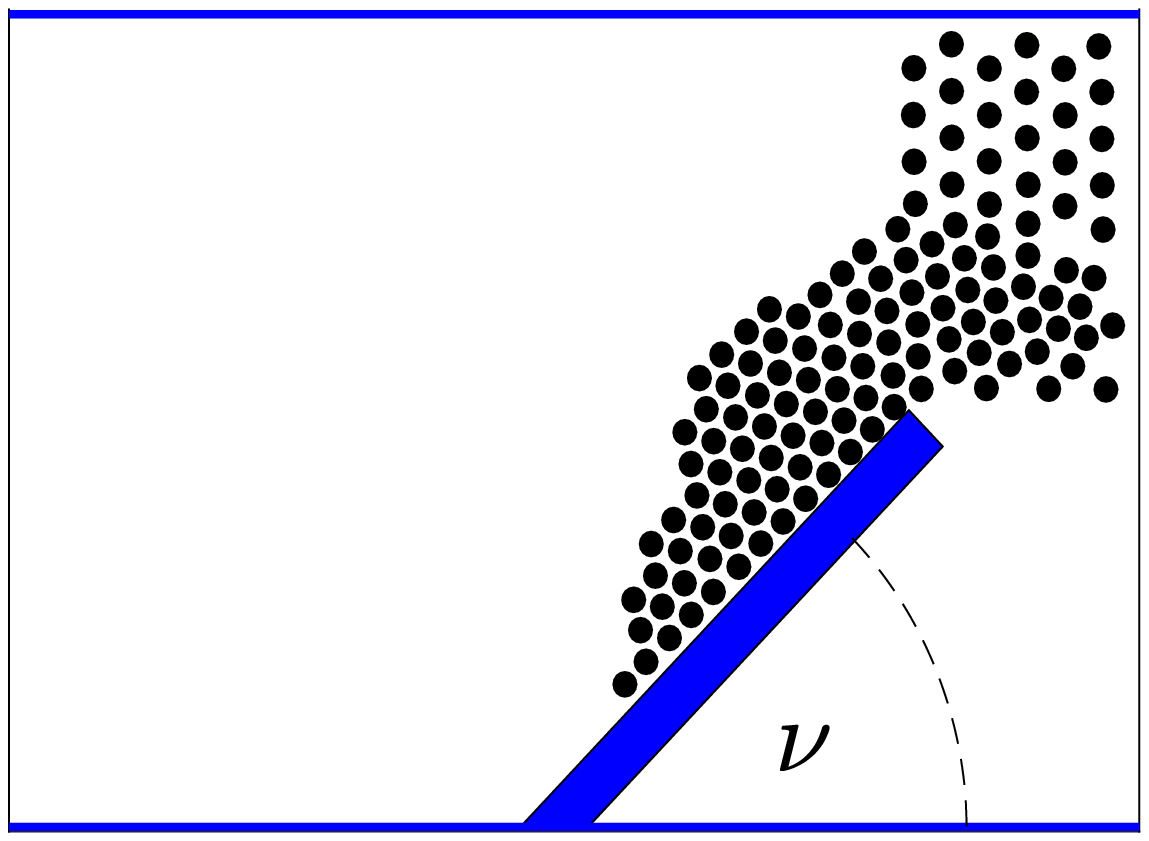}
			\subcaption{Experimental data.}
			\label{img_experimentaldata}
		\end{minipage}
		\begin{minipage}[c]{0.45\textwidth}
			\includegraphics[width=1\textwidth]{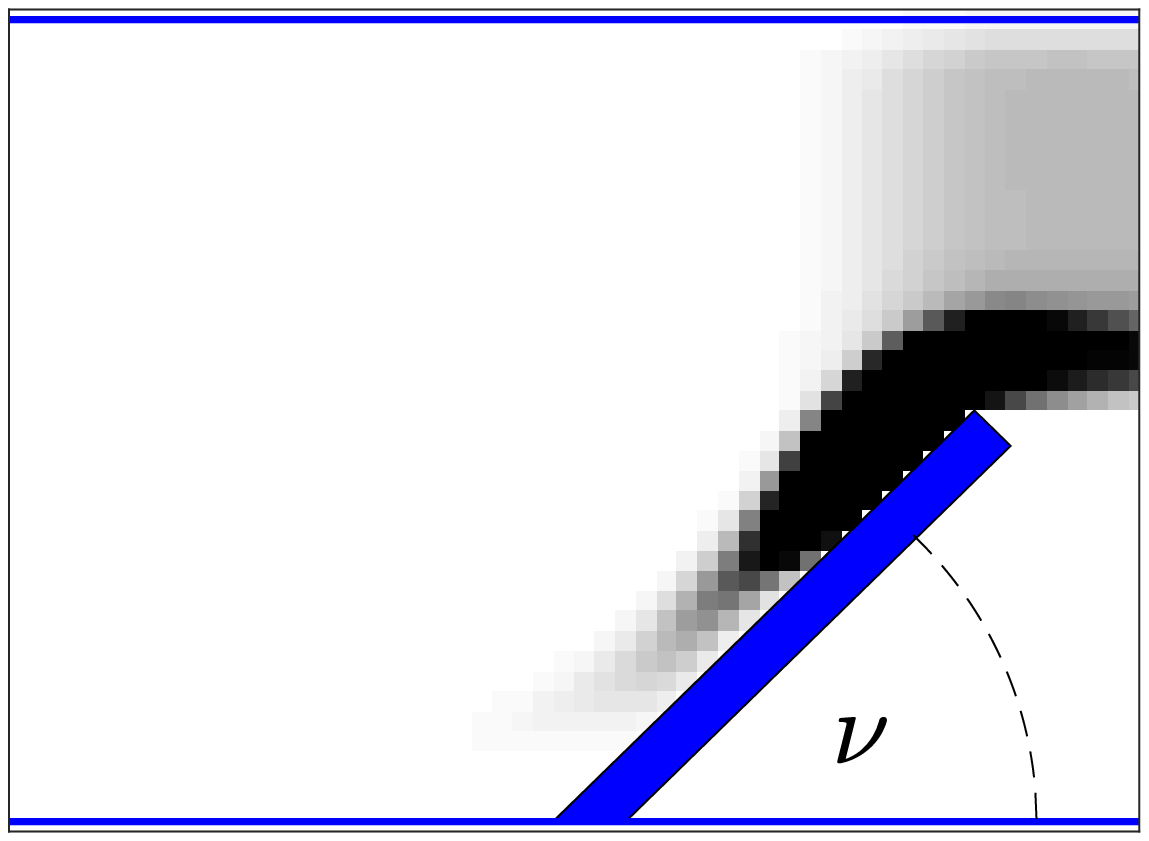}
			\subcaption{Advection-Diffusion ($H_{\xi,1}$) .}
			\label{img:density_ad_realdata}
		\end{minipage}
		\begin{minipage}[c]{0.45\textwidth}
			\includegraphics[width=1\textwidth]{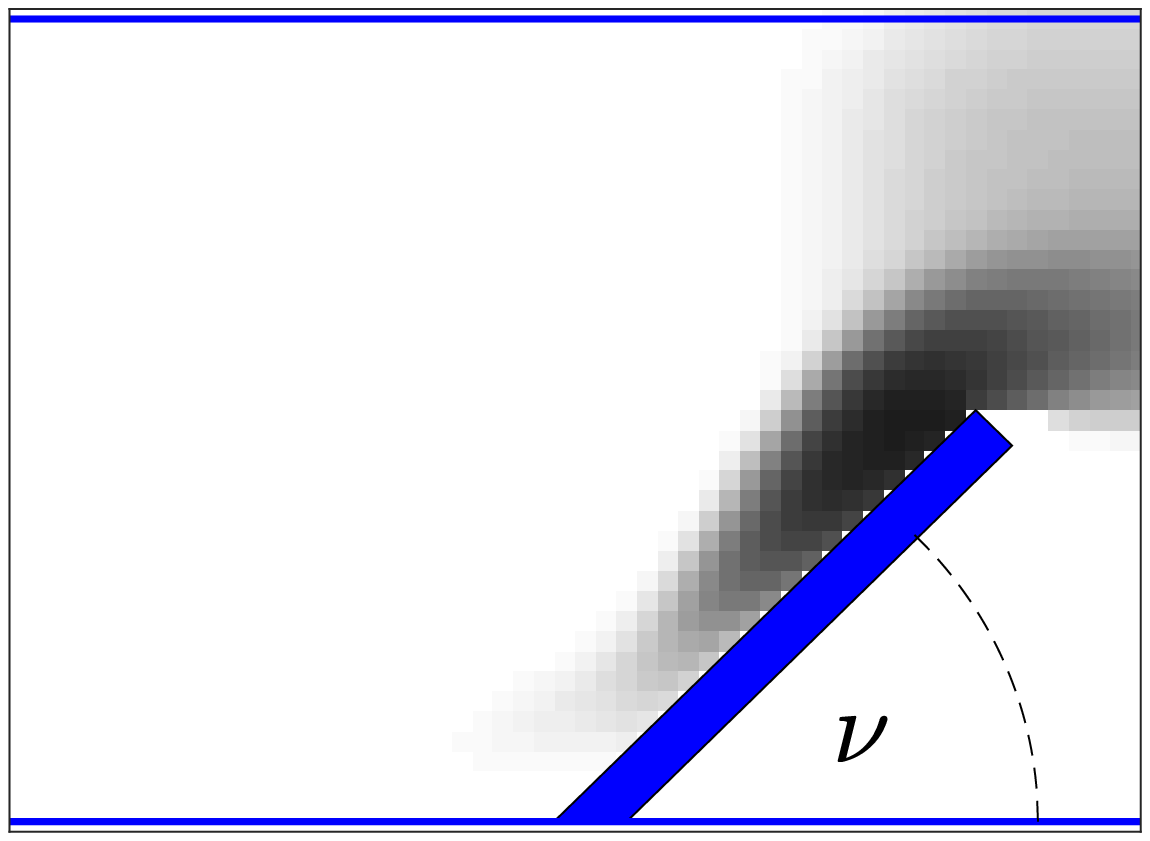}
			\subcaption{Advection-Diffusion ($H_{tan}$) .}
			\label{img:density_ad_realdata_Htan}
		\end{minipage}
		\begin{minipage}[c]{0.45\textwidth}
			\includegraphics[width=1\textwidth]{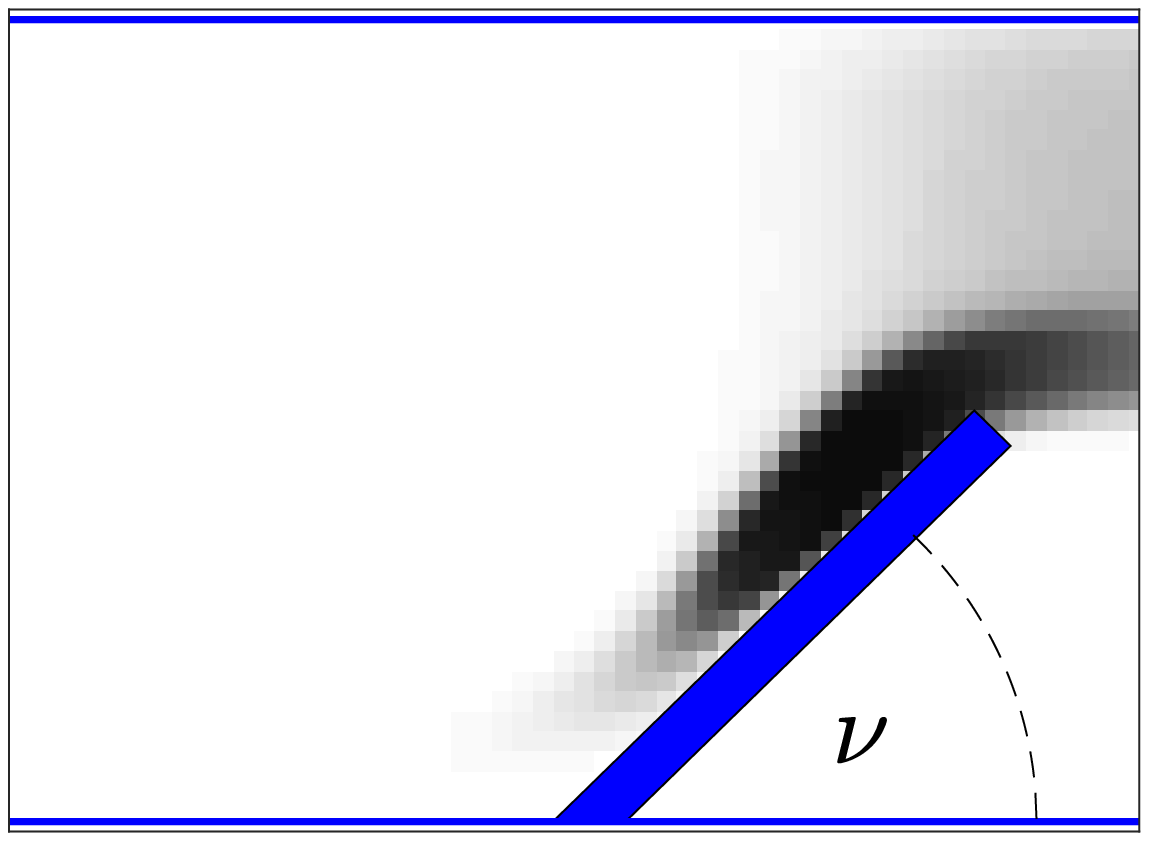}
			\subcaption{Non-Local ($H_{tan}$).}
			\label{img:density_roe_realdata}
		\end{minipage}
		\caption{Real data and density plots of the solutions $t=1.5$s.}
		\label{img:densitycomparison_realdata}
	\end{figure}

	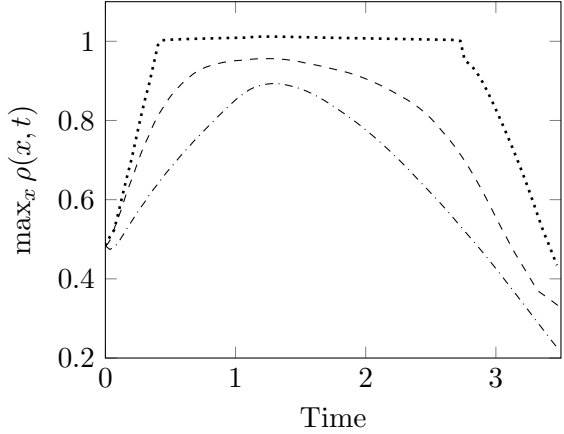
\begin{figure}[htb]
		\centering 
		\begin{tikzpicture}
			
			\begin{axis}[height=2.5cm, axis lines=none, legend columns=3, xmin=0, xmax=1, ymin=0, ymax=1, legend style={ font=\footnotesize, legend cell align=left, align = center, draw=white!15!black}] 
				\addplot[black,dotted, line width=1.0pt] coordinates {(-1,-1)};
				\addplot[black,dashdotted, line width=1.0pt] coordinates {(-1,-1)};
				\addplot[black,dashed] coordinates {(-1,-1)};
				\addplot[black,dotted,line width=1.0pt] coordinates {(-1,-1)};
				\addlegendentry{ \footnotesize Advection-Diffusion  ($H_{\xi,1}$) };
				\addlegendentry{ \footnotesize Advection-Diffusion  ($H_{tan}$) };
				\addlegendentry{ Non-Local ($H_{tan}$)};
			\end{axis}
		\end{tikzpicture}
		
		\vspace{5pt}
		\begin{subfigure}[b]{0.45\textwidth}
%
%
\begin{tikzpicture}
\setlength\fwidth{0.9\textwidth}
\begin{axis}[%
width=0.951\fwidth,
height=0.75\fwidth,
at={(0\fwidth,0\fwidth)},
scale only axis,
xmin=0,
xmax=3.5,
xlabel={Time},
ymin=0.2,
ymax=1.1,
ylabel={$\max_x \rho(x,t)$},
axis background/.style={fill=white},
legend style={at={(0.5,0.03)}, anchor=south, legend cell align=left, align=left, draw=white!10!black}
]
\addplot [color=black, dashed]
  table[row sep=crcr]{%
0	0.483520026326851\\
0.0104122295868616	0.486604264502057\\
0.021771025499802	0.489960509371458\\
0.0421222015104861	0.501297240721817\\
0.0449619004887212	0.504073126577189\\
0.0634199438472489	0.522140232950448\\
0.0809314208796983	0.538413611152928\\
0.0989161810751868	0.554281116502608\\
0.11974064024891	0.572586642833923\\
0.132992568814007	0.586815367504904\\
0.14813763003126	0.602154281851773\\
0.164702540737631	0.618052373894778\\
0.18268730093312	0.634451414304732\\
0.201145344291648	0.650459781441451\\
0.221969803465371	0.668448357103563\\
0.235695015193507	0.682223017306189\\
0.250366793247721	0.696092959373727\\
0.265985137628014	0.710002869163185\\
0.281603482008307	0.723105202115685\\
0.297695109551638	0.735811934622613\\
0.314260020258009	0.748098808959521\\
0.33082493096438	0.75962740894059\\
0.336977612083889	0.764126232705882\\
0.349756257485947	0.774690835455036\\
0.363008186051044	0.784845382444376\\
0.376260114616141	0.79423719283867\\
0.389985326344277	0.803221320603453\\
0.404183821235452	0.811791926767165\\
0.419328882452705	0.82021467939658\\
0.435420509996037	0.828462727732209\\
0.460504517637113	0.841118621635821\\
0.47375644620221	0.848278954105135\\
0.487954941093385	0.855252729113928\\
0.503573285473678	0.862248928330471\\
0.521084762506127	0.869414356656123\\
0.540016089027694	0.876483111473042\\
0.559893981875339	0.883239814044838\\
0.568886361973083	0.886444149118209\\
0.583558140027297	0.891705189616042\\
0.599649767570629	0.896809505207135\\
0.618107810929157	0.902007635923208\\
0.65076434917886	0.910988792383517\\
0.666382693559152	0.91535021202605\\
0.683420887428563	0.919456830770255\\
0.701405647624051	0.923161531309138\\
0.720810257308657	0.926536422056356\\
0.738795017504145	0.929692331341442\\
0.758199627188751	0.932819820876346\\
0.779024086362475	0.93555982262805\\
0.801741678188355	0.937932333555044\\
0.813573757264334	0.939321906768503\\
0.836291349090215	0.9416111265873\\
0.861375356731291	0.943515135489271\\
0.913436504665599	0.947183171255666\\
0.94183349444795	0.948906995613644\\
0.974963315860692	0.950292693652585\\
1.02749774695804	0.951752188506393\\
1.07151308112068	0.952674164390112\\
1.08239859387058	0.953133885032892\\
1.12499407854411	0.954820811712508\\
1.16522314740244	0.955802945380364\\
1.20545221626077	0.956173032748357\\
1.30200198152076	0.956049945722365\\
1.34980358098771	0.954976960581968\\
1.39949831310683	0.95325102047754\\
1.4487197620629	0.950937780757088\\
1.4969946446929	0.948057280842855\\
1.54479624415985	0.944595274432217\\
1.56656726965966	0.943180563160859\\
1.62288796606132	0.939703091873277\\
1.68062851195209	0.93552205476207\\
1.85668984860267	0.921542127308624\\
1.90496473123266	0.916711640167749\\
1.95276633069962	0.91128835436818\\
2.00151449649265	0.905142264160801\\
2.0516825117748	0.89820279189011\\
2.10421694287215	0.890317685830255\\
2.15533152448038	0.882034416165149\\
2.20076670813214	0.874074461979628\\
2.32098063154409	0.851808415113207\\
2.35600358560899	0.844465467510784\\
2.38771355753262	0.837206268459193\\
2.41753039680408	0.829761339749295\\
2.44592738658643	0.822041075484351\\
2.47290452687967	0.814079846569801\\
2.49893510084682	0.805771814425786\\
2.52449239165094	0.79697735036735\\
2.54957639929201	0.787699077353026\\
2.57418712377005	0.777947203139035\\
2.59217188396554	0.770890794052885\\
2.61962230742181	0.760672809003455\\
2.64707273087808	0.749797528696214\\
2.67404987117131	0.738459315525792\\
2.70150029462759	0.726256157972259\\
2.72895071808386	0.713383520799844\\
2.75687442470317	0.699611669150467\\
2.78527141448552	0.684923727190555\\
2.81461497059395	0.669050823236742\\
2.84490509302845	0.651964208166978\\
2.87661506495208	0.633367797956528\\
2.91021816952786	0.612946248132394\\
2.94666097308188	0.590071977806278\\
2.98736332510324	0.563783089404626\\
3.04179088885275	0.528538180632546\\
3.12792842452588	0.475741931946023\\
3.14118035309098	0.467948629522144\\
3.32055467188282	0.370030479527736\\
3.38539446521919	0.354971064988086\\
3.45023425855555	0.339283737622398\\
3.49992899067467	0.326844811404945\\
};

\addplot [color=black, dotted,line width=1.0pt]
  table[row sep=crcr]{%
0	0.483520026326851\\
0.0336717514850737	0.50610642895934\\
0.0673435029701475	0.523343589508454\\
0.101015254455221	0.567982803990923\\
0.134687005940295	0.615164361155777\\
0.168358757425369	0.660673271694228\\
0.202030508910442	0.703588384795174\\
0.269374011880589	0.804625805495188\\
0.303045763365663	0.851050647556623\\
0.336717514850737	0.892894845756955\\
0.404061017820884	0.990809381485593\\
0.437732769305958	1.00257194745422\\
0.505076272276106	1.00418343771242\\
0.673435029701474	1.00621683329269\\
0.808122035641769	1.00728074891758\\
0.90913729009699	1.00818520064903\\
1.07749604752236	1.00926037587247\\
1.11116779900743	1.01007850682309\\
1.14483955049251	1.01139499599977\\
1.21218305346265	1.01187501805418\\
1.31319830791787	1.01175370773069\\
1.48155706534324	1.01072211264269\\
1.68358757425368	1.00921304419037\\
2.66006836732082	1.00358356090982\\
2.6937401188059	1.00292029830447\\
2.72741187029097	1.00072138977364\\
2.76108362177604	0.951185503389304\\
2.79475537326112	0.943365753855244\\
2.82842712474619	0.931754623080617\\
2.86209887623126	0.916415504843685\\
2.89577062771634	0.897882712005315\\
2.92944237920141	0.876884524525984\\
2.96311413068649	0.854121950241635\\
2.99678588217156	0.83016364314554\\
3.03045763365663	0.805412610726459\\
3.06412938514171	0.780091188803888\\
3.09780113662678	0.754244560410453\\
3.13147288811185	0.727777660611414\\
3.16514463959693	0.700518594512651\\
3.198816391082	0.672295217156005\\
3.23248814256708	0.643001079194012\\
3.26615989405215	0.612627312881441\\
3.29983164553722	0.58126023203558\\
3.3335033970223	0.549065191309357\\
3.36717514850737	0.517726239908837\\
3.40084689999244	0.491363526705092\\
3.43451865147752	0.463411207507343\\
3.46819040296259	0.434380936341711\\
};

\addplot [color=black, dashdotted]
table[row sep=crcr]{%
	0	0.483520026326851\\
	0.0336717514850737	0.474403537020665\\
	0.0673435029701475	0.479214521829266\\
	0.101015254455221	0.492654061702843\\
	0.134687005940295	0.510072262789374\\
	0.168358757425369	0.528806594499292\\
	0.202030508910442	0.545225044645564\\
	0.235702260395516	0.562673885577584\\
	0.269374011880589	0.579142651106948\\
	0.303045763365663	0.594326324506714\\
	0.336717514850737	0.611467580145321\\
	0.370389266335811	0.626538483287752\\
	0.404061017820884	0.639929052528486\\
	0.437732769305958	0.655010703833062\\
	0.471404520791032	0.669201953734046\\
	0.538748023761179	0.69502900750537\\
	0.572419775246253	0.708587888681238\\
	0.606091526731327	0.721012701166304\\
	0.6397632782164	0.732512349542889\\
	0.673435029701474	0.745444787416865\\
	0.707106781186547	0.757294623120107\\
	0.740778532671621	0.768099654385573\\
	0.774450284156695	0.779914828201884\\
	0.808122035641769	0.790904256693116\\
	0.841793787126842	0.802575131929121\\
	0.942809041582064	0.834850947562293\\
	1.01015254455221	0.855004626824334\\
	1.07749604752236	0.871608225116459\\
	1.11116779900743	0.878462091966138\\
	1.14483955049251	0.883764783706209\\
	1.17851130197758	0.887636950183598\\
	1.21218305346265	0.890195258581931\\
	1.24585480494773	0.892026223980768\\
	1.2795265564328	0.893123706099361\\
	1.31319830791787	0.89316800494679\\
	1.34687005940295	0.892252555341746\\
	1.38054181088802	0.890455639272931\\
	1.44788531385817	0.885224203096438\\
	1.48155706534324	0.881975949478357\\
	1.51522881682832	0.878080335159967\\
	1.54890056831339	0.873580609145499\\
	1.58257231979846	0.868516799173972\\
	1.61624407128354	0.862926445488076\\
	1.64991582276861	0.856844932233157\\
	1.68358757425368	0.85030550268705\\
	1.71725932573876	0.843339059046846\\
	1.85194633167905	0.813644700385273\\
	1.9192898346492	0.797321449489489\\
	1.98663333761935	0.779946568057703\\
	2.0539768405895	0.761575997704151\\
	2.12132034355964	0.74222588324225\\
	2.18866384652979	0.721883634934939\\
	2.25600734949994	0.700518507116767\\
	2.32335085247009	0.678090549131207\\
	2.5590531128656	0.596415564815285\\
	2.62639661583575	0.571696401429044\\
	2.6937401188059	0.546020254743533\\
	2.72741187029097	0.532827456796406\\
	2.82842712474619	0.494611772760583\\
	2.89577062771634	0.46828186859042\\
	2.96311413068649	0.441185947436002\\
	3.03045763365663	0.413403368150947\\
	3.13147288811185	0.370694794023075\\
	3.43451865147752	0.241134150764409\\
	3.46819040296259	0.227313291281734\\
};

\end{axis}
\end{tikzpicture}%
		\end{subfigure}
		\caption{Maximum density over time.}
		\label{img:maxdens_realdata}
	\end{figure}

	
	\subsection{Numerical results for the swarming model} \label{sec:numerics_swarming}
	
	\subsubsection{Scattering a single swarm at a boundary}
	
	In~\cite{ArmMarTha2017} the collision of flocks with walls has been analyzed for the microscopic  attraction-repulsion model~\eqref{eq:attraction_repulsion}. Setting $F(x_i- x_j) =\lambda \nabla_{x_i} U(x_i -x_j)$  and varying $\lambda >0$ changes the relative strength of the damping and the swarming potential. Reflecting the individual's velocity specularly at the wall,  different reflection patterns for the swarm have been obtained. Specifically, reflection laws that show the outgoing angle of the flock as a function of the incoming angle $\theta^0$ for different scalings of the interaction potential $\lambda$ have been determined. 
	In the following, we study the boundary behaviour for swarms and for varying diffusion coefficients $C$ to validate the macroscopic model~\eqref{eq:macroscopicswarming_singleswarm}. 
	
	We write the initial velocity with initial heading $\theta^0$ as  $\overline{v}^0 = \left( \cos(\theta^0), \sin(\theta^0) \right)$  such that $\norm{\bar{v}^0} = 1$ and  set the initial density to $\rho_0 = \rhocrit\, \mathbbm{1}_{ \left( \norm{x}  \leq r \right)}, r>0,$ representing a flock  centered at $ x =(0,0)^T$. The boundary  is placed at $x^{(1)} = 0.2$ with $\vec{n} = (1,0)^T$ (see Fig.~\ref{img:Incomingvsoutgoingangle}). The critical density is set to $\rhocrit = 1$. More details of the simulation can be found in Appendix \ref{app:swarm}.

	Figure~\ref{img:SAW_initialconfiguration} shows the initial configuration of a flock which is directed towards the boundary with $\theta^0 = 45$.  We  introduce $\delta = C /v^{(1),0}$ to evaluate the boundary interaction of the swarm for different values of $\delta \geq 1$. Thus diffusion in regions where the critical density is exceeded increases with $\delta$.  The contour of the swarm density 
	after the boundary interaction (at $t=0.3$) is  obtained with Algorithm~\ref{alg:swarming_boundary} (Appendix \ref{app:swarm})  and shown in Figure~\ref{img:SAW_aftercollision}\footnote{The swarm profile slightly enlarges over time due to numerical diffusion.}. The arrows depict the direction of the flock after the interaction with the wall.

	\begin{figure}[htp]
		\begin{subfigure}{0.45\textwidth}
			
			\centering
			\begin{tikzpicture}
				
				\coordinate (O) at (0,0) ;
				\coordinate (A) at (0,2) ;
				\coordinate (B) at (0,-2) ;
				
				\node[right] at (0.2,0) {Wall};
				\node[left] at (-1,-1) {};
				
				\draw[] (A) -- (B) ;
				\draw[dash pattern=on5pt off3pt] (0,0) -- (-1.5,0) ;
				\draw[dash pattern=on5pt off3pt] (0,-1.44) -- (-1.5,-1.44) ;
				
				\draw[black,thick] (O) -- (135:2);
				\draw[black, thick] (O) -- (50:-1.93);
				
				\draw (-0.8,0.8) arc (135:185:1);
				\draw (-0.4,-1.44) arc (0:40:1) ;
				\node[] at (-0.7,-1.2)  {$\theta^{0}$};
				\node[] at (-0.6, 0.25)  {$\theta^{r}$};
			\end{tikzpicture}
			\caption{Incoming angle and reflection angle.}
			\label{img:Incomingvsoutgoingangle}
		\end{subfigure}
		\begin{subfigure}{0.45\textwidth}
			\includegraphics[width=\textwidth]{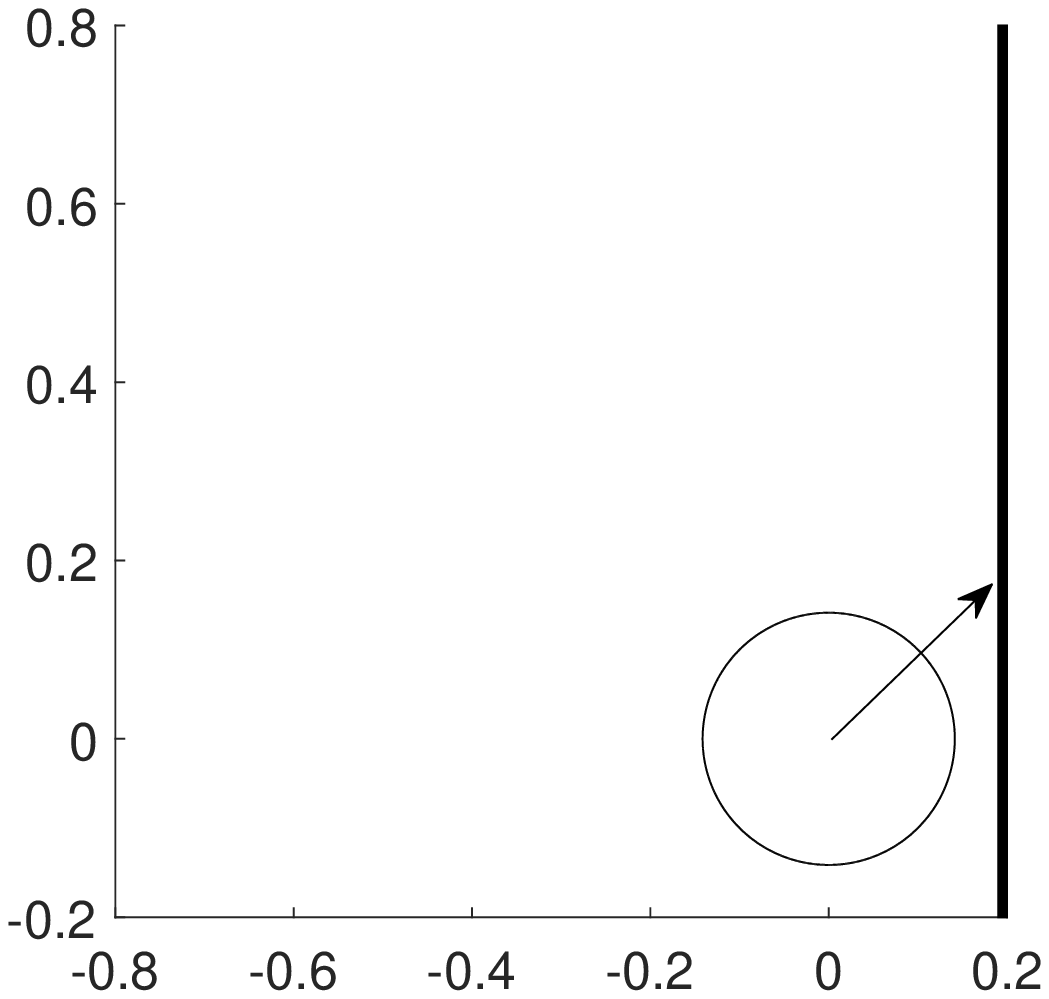}
			\caption{Initial configuration.}
			\label{img:SAW_initialconfiguration}
		\end{subfigure}\\
		\centering
		\begin{subfigure}{0.45\textwidth}
			\includegraphics[width=\textwidth]{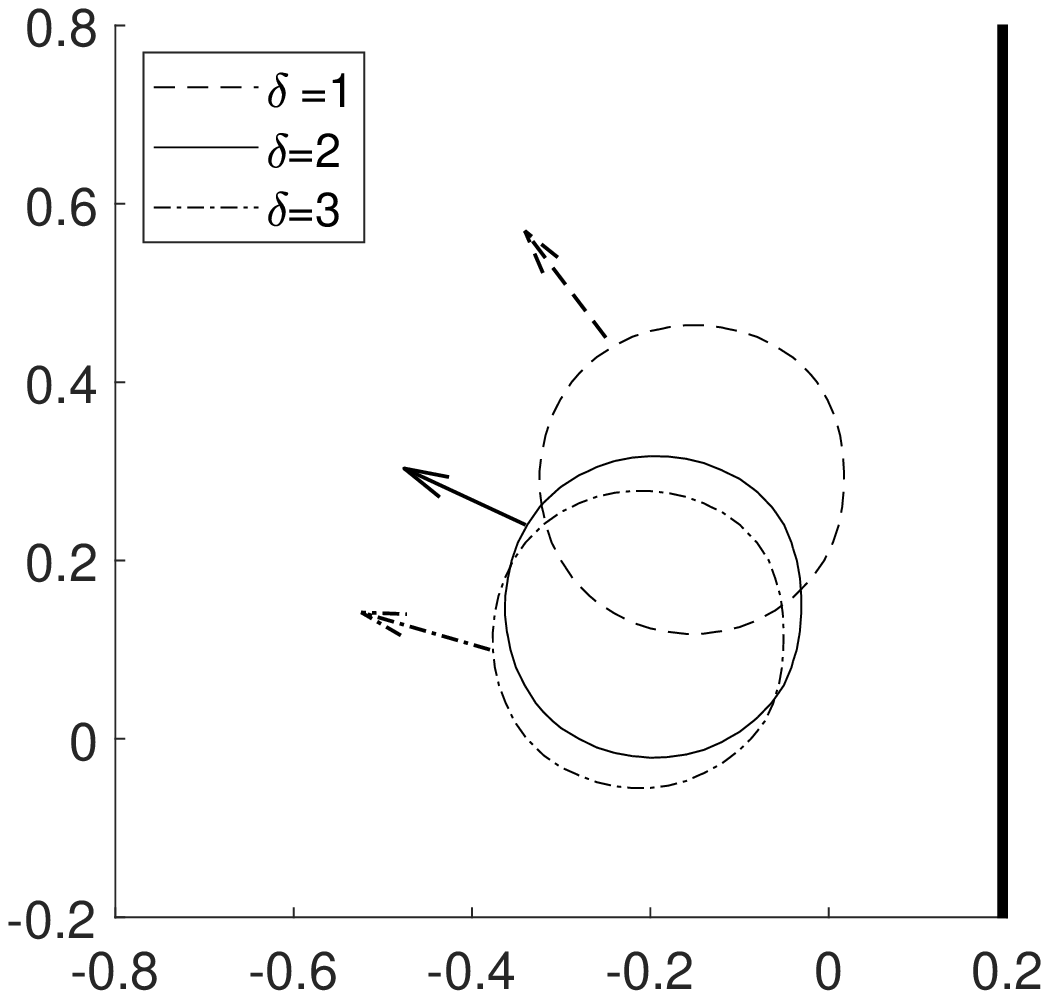}
			\caption{Swarm contours  after collision for different $\delta$.}
			\label{img:SAW_aftercollision}
		\end{subfigure}
		\caption{Collision of a swarm with a boundary.}
	\end{figure}

	The reflection angle $\theta^{r}$ for increasing values of $\delta$ and angles $\theta^0 \in \lbrace 30, 45, 60 \rbrace$ are displayed in Figure~\ref{img:SAW_304560deg_scaledelta}. We find that for increasing $\theta^0$, the reflection angle $\theta^r$ increases,  matching the microscopic observations in the attraction-repulsion model~\cite{ArmMarTha2017}. For small diffusion, the swarm aligns with the wall. Increasing diffusion, the outgoing angles decrease  and the swarms are reflected away from the wall. The asterisk in Figure~\ref{img:SAW_304560deg_scaledelta} marks the diffusion constant for which the flock as a whole reflects specularly, i.e.  $\theta^0 = \theta^r$. The results of Table~\ref{tab:SAW_aftercollison} are marked with a circle in Figure~\ref{img:SAW_304560deg_scaledelta}.
	
	\begin{figure}[htp]
		\centering
		\includegraphics[width=0.65\textwidth]{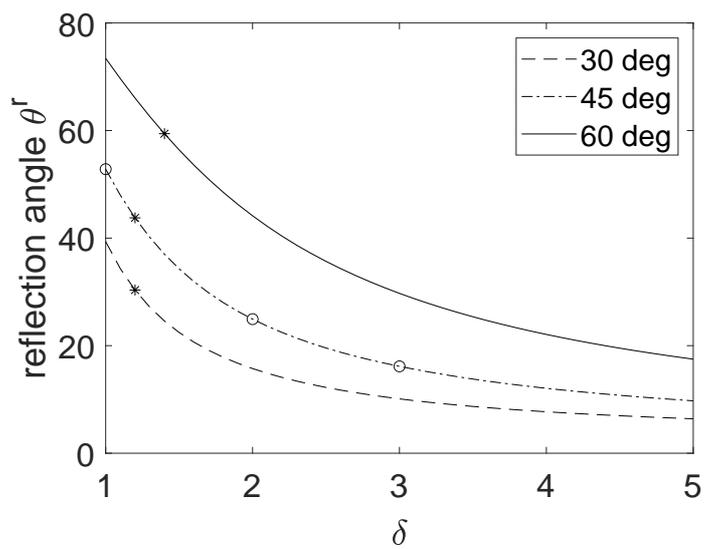}
		\caption{Reflection angle for $\theta^0=$ 30, 45 and 60 deg.}
		\label{img:SAW_304560deg_scaledelta}
	\end{figure}
	
	\begin{table}[htp]
		\centering
		\caption{Reflection angle of the experiments in Figure~\ref{img:SAW_aftercollision}.}
		\label{tab:SAW_aftercollison}
		\begin{tabular}{c | c c c  }
			$\delta$ &	1 & 2 & 3 \\
			\hline
			$\theta^r$ & 52.81 & 24.93 & 16.15
		\end{tabular}
	\end{table}

	\subsubsection{Collision of several swarms}

	We consider two swarms described by the density vector $\boldsymbol{\rho} = (\rho_1,  \rho_2)$, where the density of swarm $i$ is given by $\rho_i(x,t), ~i = 1,2$.  Assuming that each swarm $i$  has an identity and adjusts its velocity according to the mean velocity of its members, while diffusion is driven by the total density taking into consideration crowding from the other swarm, we obtain the system of non-linear advection-diffusion equations
	\begin{equation*}
		\partial_t \rho_i + \nabla \cdot \left(\rho_i \overline{v_i} - \rho_i C \, H\left( \sum_{i=1}^{2} \rho_i   > \rhocrit \right) \nabla \left(\sum_{i=1}^{2} \rho_i \right)  \right) = 0,
	\end{equation*}
	in two space dimensions. Let 
	\begin{align*}
		f_i(\rho) = \rho \bar{v}_i, \qquad g_i(\boldsymbol{\rho}) = \rho_i C \mathbbm{1}_{ \left( \Phi(\boldsymbol{\rho})   > \rhocrit \right)}, \qquad \Phi(\boldsymbol{\rho}) = \sum_{z=i}^{2} \rho_i.
	\end{align*}
	We can reformulate this equation in the more general form as  
	
	\begin{equation} \label{eq:swarming_multiclass}
		\partial_t \rho_i + \nabla \cdot \left( f_i(\rho_i)- g_i(\boldsymbol{\rho}) \nabla \Phi(\boldsymbol{\rho})  \right) = 0,
	\end{equation}
	frequently used in batch settling and sedimentation processes~\cite{BerBueKar2003, BueDieCar2017}. \\
	
	We consider the collision of two swarms heading in opposite directions  discussed in~\cite{ArmMarTha2017} (see Figure~\ref{img:collisionofswarms_initialpos}). The initial velocities are $\overline{v_L}^0 = \left( \cos(45), \sin(45) \right)^T$ and $\overline{v_R}^0 = \left( -\cos(45), \sin(45) \right)^T$ for the left and the right swarm, respectively. We set $\rhocrit =1$. Initially, the swarm densities are set to $\rho_{L}^0 = \rho_{R}^0= 0.8~\rhocrit$, such that diffusion is activated when the swarms collide. Note that for $\rho^0 <0.5$, the two swarms just pass each other since their cumulated density is always below $\rhocrit$ in agreement with 
	~\cite{ArmMarTha2017}.

	Using Algorithm~\ref{alg:swarming_boundary}  (Appendix \ref{app:swarm}) with $\dx^{(1)} =\dx^{(2)} = 5 \cdot 10^{-2}$,  we compute the approximate solution to equation~\eqref{eq:swarming_multiclass}. 
	Initial positioning and simulation results are displayed in Figure~\ref{img:collisionofswarms}. For the lower diffusion coefficient $C=0.1$, we observe that the two swarms merge, interact with each other and after the interaction continue on separate paths. After the interaction, the density inside of the swarms is $\einhalb \rhocrit$. In contrast, for the diffusion coefficient $C=2$, the two swarms are still merged\footnote{We observe a small drift of the density in Figure~\ref{img:collision_C2} to the left because we compute first the left swarm and afterwards the right swarm in each time step.}. The density inside the new swarm is approximately $\rhocrit$. 
	
	\begin{figure}[htp] 
		\centering
		\begin{subfigure}{0.32\textwidth}
			\centering
			\includegraphics[width=1\textwidth]{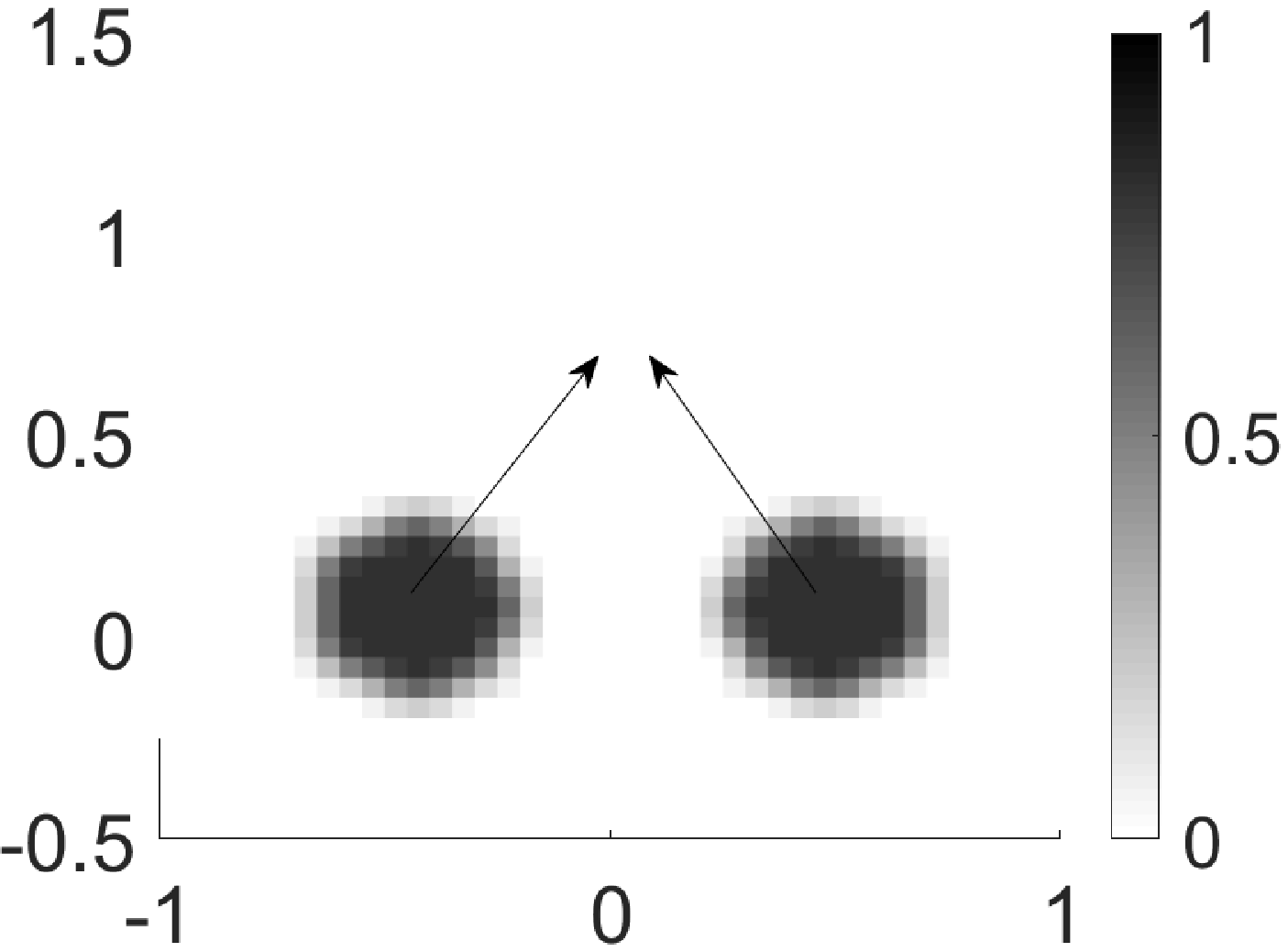}
			\subcaption{Initial positions.}
			\label{img:collisionofswarms_initialpos}
		\end{subfigure}
		\begin{subfigure}{0.32\textwidth}
			\centering
			\includegraphics[width=1\textwidth]{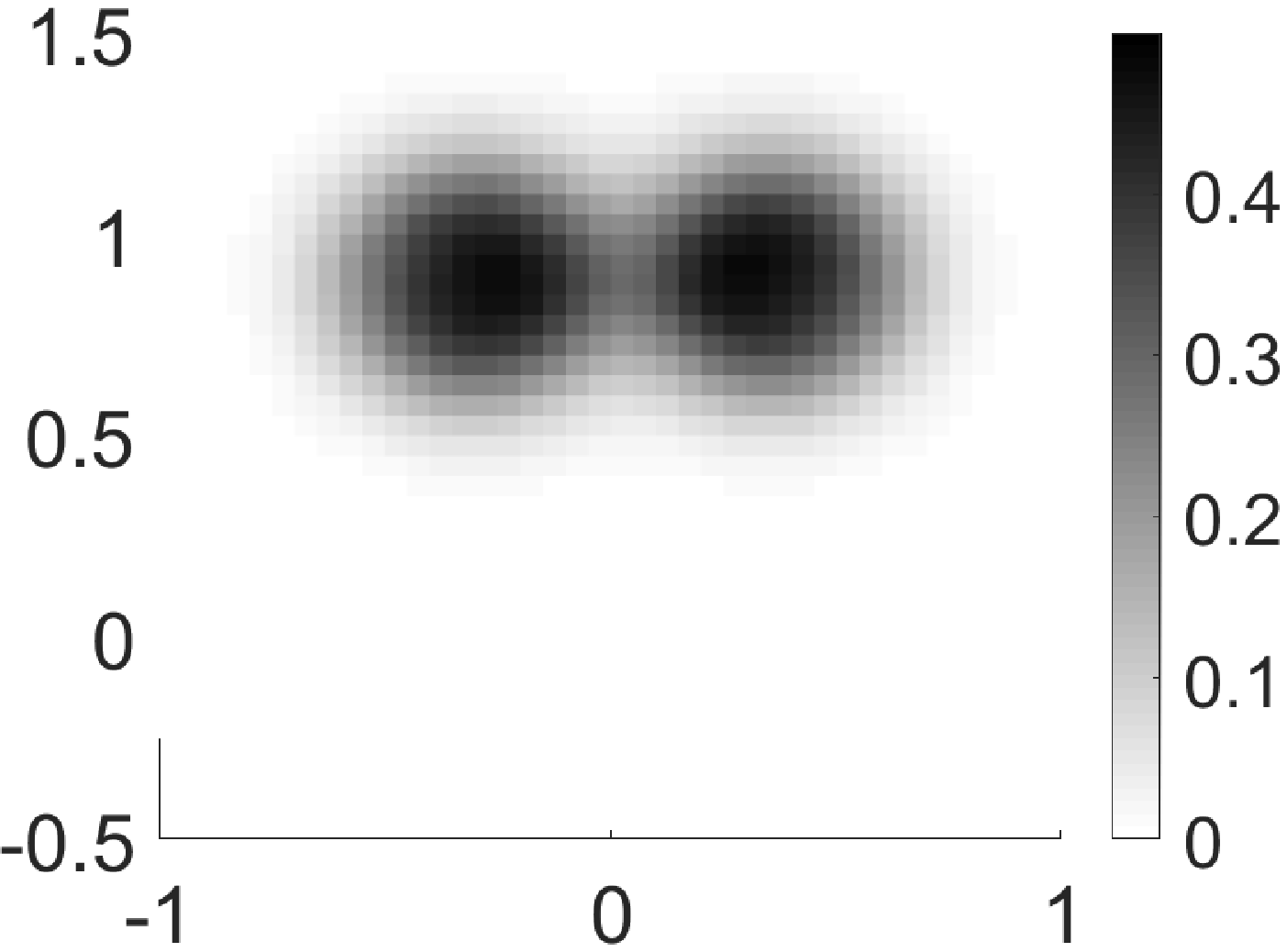}
			\subcaption{$C =0.1$ at $t = 1.2$.}
		\end{subfigure}
		\begin{subfigure}{0.32\textwidth}
			\centering
			\includegraphics[width=1\textwidth]{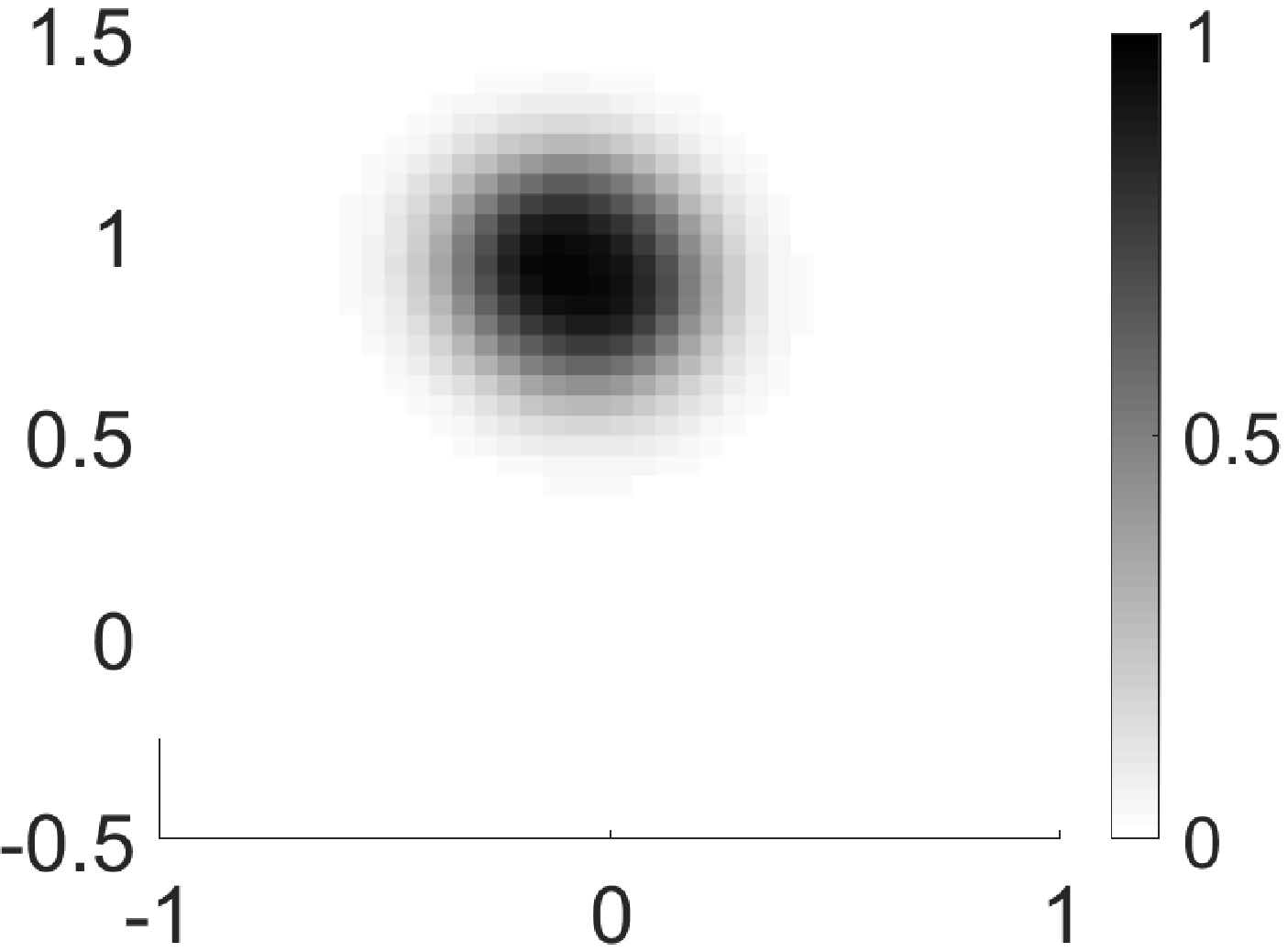}
			\subcaption{$C = 2$ at $t = 1.2$.}
			\label{img:collision_C2}
		\end{subfigure}
		\caption{Collision of swarms.}
		\label{img:collisionofswarms}
	\end{figure}
	
	
	\section{Conclusion}\label{sec:conclusion}
	
	We have shown how locally repelling forces modeling geometric exclusion principles at the microscopic level lead to a discontinuous 
	advection-diffusion model at the macroscopic limit for very generic setups of interacting particle flows. 
	This closes an important gap in the theoretical program to develop macroscopic models that describe emergent phenomena of large ensembles of interacting particles as limits of their microscopic behavior. Specifically, such geometric exclusion applies to boundaries and obstacles whose macroscopic representation have typically been done by ad-hoc boundary conditions. 
	
	The resulting transport equation is a generalization of ~\cite{GoeKlaTiw2015}. It is
	strictly hyperbolic as long as the density of individuals is below a critical density and becomes parabolic when the density exceeds that threshold. The parabolic part is then described by a density dependent diffusion coefficient. We introduced the operator splitting algorithm to solve the hybrid partial differential equation splitting the advective from the diffusive part of the equation and thus being able to handle the discontinuous change in type. 
	
	Applications are far reaching: Material flows on conveyor belts fit that model as an undisturbed conveyor belt will transport all parts on it with the same velocity without any relative motion. Interacting with deflectors or obstacles will push the parts together such that they collide and create relative movement that can best be described by diffusion \cite{GoeHohSch2014}.  We show that numerical solutions of the degenerate advection-diffusion model for conveyor belt flows around a deflector reproduce quantitatively the experimental results for similar situations. 
	
	In another application, the attraction repulsion model \cite{ChuDorMar2007} of self-propelling particles generates flock solutions where the forcing potential reaches a minimum and all particles move with the same velocity and a crystal-like fixed relative position. 
	Again, relative motion is introduced
	when the particles become close enough for collisions which macroscopically is depicted by exceedance of the critical density. The macroscopic model successfully replicated microscopic reflection laws and high-impact vs. low-impact scattering depending 
	on the relative strength of the potential forcing vs. the damping forces, reported in \cite{ArmMarTha2017}.
	

	\section*{Acknowledgments} The authors are grateful for the support of their joint research by the DAAD (Project-ID 57444394). J. Weissen and S. G\"ottlich are supported by the DFG project GO 1920/7-1. Further, the authors would like to thank Stephan Knapp for valuable discussions and  helpful suggestions which contributed to the emergence of the present work.

	\appendix
	\section{Simulation of the non-local model} \label{app:non-local}
	The space step for our simulations discussed in section \ref{sec:non-local} is $\dx^{(1)} =\dx^{(2)} = 10^{-2}$. We consider numerical results for equation~\eqref{eq:macrolimit_materialflow_totallynonlocal} computed by the finite volume Roe scheme, see~\cite{GoeHohSch2014}. We set $\rhomax = 1$, $\epsilon = 2v^{(1)}_T$ and the mollifier $\eta$ is
	\begin{align*}
		\eta(x) &= \frac{\sigma}{2 \pi } e^{-1/2 \sigma \norm{x}_2^2}, 
	\end{align*}
	with $\sigma = 10^4$. The CFL time step for the Roe scheme is $\Delta t = 4.7328 \cdot 10^{-4}$~\cite{RosWeiGoa2019}. The CFL step is relatively small due the approximation of the Heaviside function
	\begin{align} \label{eq:Heaviside_atan}
		H_{tan}(u) &= \frac{\arctan(50 (u-\rhomax))}{\pi} + \frac{1}{2},
	\end{align}
	in the numerical flux function of the Roe scheme. Sharper approximations would further strengthen the time step size.
	
	For comparability, we compute the advection-diffusion equation once with the Heaviside approximation~\eqref{eq:Heaviside_atan} and once with the sharper approximation~\eqref{eq:Heaviside_xi}.
	We fix $C:= \bar{C} / \gamma_b$ in the diffusion to $C=2 v^{(1)}_T$. We treat the deflector as internal boundary $\partial \Omega \subset \Omega$ in the advection-diffusion equation. We use the operator splitting method~\eqref{eq:Upwind}-\eqref{eq:implicitscheme} with the boundary conditions at the deflector from Section~\ref{sec:boundaryconditions}. The CFL time step~\eqref{eq:CFL_operatorsplitting} is $\Delta t = 7.32 \cdot 10^{-2}$. Note that the application of the operator splitting allows to compute with larger time steps because the implicit method is used to compute the parabolic part. However, in comparison to the explicit Roe scheme, a system of nonlinear equations has to be solved in each iteration.

	\section{Simulation of swarms} \label{app:swarm}
	\begin{algorithm}[h]
		\caption{Numerical simulation of a swarm in a bounded domain}\label{alg:swarming_boundary}
		\begin{algorithmic}[1]
			\REQUIRE Domain with boundary $\Omega \cup \partial \Omega$,  initial conditions $\rho_0, \overline{v}^0$ for $x \in \Omega \cup \partial \Omega$, diffusion coefficient $C$ and critical density $\rhocrit$, step sizes $\dx^{(1)}, \dx^{(2)}, \Delta t$ 
			\ENSURE Densities vectors $\rho^s=(\rho_{ij}^s)_{i \in \{1,\dots, N_{x^{(1)}} \}, j \in \{1,\dots, N_{x^{(2)}} \}}$
			\STATE Set $s=0, t^s = 0$
			\WHILE {$t^s <T$}
			\STATE Set $\Delta t$ according to the CFL condition~\eqref{eq:CFL_operatorsplitting} and set $s=s+1, t^s = t^{s-1} + \Delta t$
			\FOR {$i = 1,\dots, N_{x^{(1)}}$} 
			\FOR { $j = 1,\dots, N_{x^{(2)}}$}
			\STATE Compute the velocity $v_{ij}^s$~\eqref{eq:calcvelocity}.
			\STATE If $x_{i+1j}, x_{ij+1}, x_{i-1j}$ or $x_{ij-1} \in \partial \Omega$ and $\langle v_{ij}^{s}, \vec{n} \rangle > 0$, apply specular reflection
			$$
			v_{ij,new} = v_{ij}^s - 2 \langle v_{ij}^s , \vec{n} \rangle \vec{n},
			$$
			where $\vec{n} $ is the outer normal vector at the boundary and update $v_{ij}^s = v_{ij,new}$.
			\ENDFOR
			\ENDFOR
			\STATE Compute the new average velocity $\bar{v}^s$ of the swarm~\eqref{eq:calcmeanvelocity} and compute the solution $\rho(x,t^s)$ to
			\begin{align*} \partial_t \rho + \nabla_x \cdot \left(\rho  \overline{v}^s - \rho C \nabla \rho H(\rho - \rhocrit)) \right) &= 0 &\qquad x \in \Omega \\
				\left(\rho \overline{v}^s - \rho C \nabla \rho H(\rho - \rhocrit) \right) \cdot \vec{n} &= 0 &\qquad x \in \partial \Omega \\
				\rho(x,t^{s-1}) &= \rho_{ij}^{s-1} &\qquad x \in C_{ij}
			\end{align*}
			with the operator splitting method~\eqref{eq:Upwind}-\eqref{eq:implicitscheme}.
			\ENDWHILE
		\end{algorithmic}
	\end{algorithm}

	To compute an approximate solution to~\eqref{eq:macroscopicswarming_singleswarm}, we discretize with step sizes $\dx^{(1)}, \dx^{(2)}, \Delta t$ and have to iteratively determine the velocity $\bar{v}^s, s =1, \dots N_t$. For a given location $x_{ij} \in \Omega$ and fixed time $t^s = s \Delta t$, we determine the velocity 
	
	\begin{align}
		v_{ij}^{s} &= \left( \overline{v}^{s-1} -  C \nabla \rho_{ij}^{s-1} \, H(\rho_{ij}^{s-1} - \rhocrit) \right), \label{eq:calcvelocity} 
	\end{align}
	and compute the new mean velocity as the weighted average 
	\begin{align}
		\bar{v}^s &= \sum_{(i,j)} \frac{\rho_{ij}^{s-1}}{\sum_{(i,j)} \rho_{ij}^{s-1}} v_{ij}^s, \qquad \bar{v}^s = \bar{v}^s \frac{\norm{\bar{v}^0}}{\norm{\bar{v}^s}},  \label{eq:calcmeanvelocity} 
	\end{align}
	which is normalized such that $\norm{\bar{v}^s} = \norm{v^0}$. For inner grid cells, we determine $\nabla \rho_{ij}^{s-1}$ with central differences. To determine the velocity  $v^s_{ij}$ in a cell $(i,j)$ at the boundary, we approximate the gradient $\nabla \rho_{ij}$ using neighbouring cells in transport direction as follows
	
	\begin{align*}
		\nabla \rho_{ij}^{s-1} \approx \begin{cases}
			\frac{\rho_{i+1j}^{s-1}  - \rho_{ij}^{s-1}}{\Delta x^{(1)}} &\text{ if } (i+1,j) \in B \text{ and } \bar{v}^{(1),s-1} \geq 0, \\
			\frac{\rho_{ij}^{s-1} - \rho_{i-1j}^{s-1}}{\Delta x^{(1)}} &\text{ if } (i+1,j) \in B\text{ and } \bar{v}^{(1),s-1}<0.
		\end{cases}
	\end{align*}
	
	The gradient is approximated using forward differences if the flock is moving towards the boundary, i.e., $\bar{v}^{(1),s-1} \geq 0$, to reflect the swarm from the boundary. If instead of approaching the wall, the flock moves away from the boundary, i.e., $v^{(1),s-1} <0$, the gradient is approximated with the backward difference. When the flock approaches the boundary, the flock solution breaks apart due to the boundary influence. If $ \langle v_{ij}^s , \vec{n} \rangle > 0$, we assume that the velocity is reflected specularly for an individual cell. This is in line with the microscopic treatment~\cite{ArmMarTha2017} where single individuals are reflected specularly at the boundary. In particular, the macroscopic velocity~\eqref{eq:calcvelocity} is updated as follows
	\begin{equation*}
		v_{ij,new}^{s} = v_{ij}^s - 2 \langle v_{ij}^s, \vec{n}  \rangle \vec{n},
	\end{equation*}
	before calculating the new average velocity. 
	
	To evaluate whether the density is at the level of the critical density, we use $H_{\xi, \rhocrit - \xi}(\rho_{ij}^{s-1})$, such that $H_{\xi, \rhocrit- \xi}(\rhocrit) = 1$. Below the maximum density, the transport velocity $\overline{v}^s$ is reflected specularly. When the critical density $\rho_{ij}^{s-1}=\rhocrit$ is reached in the cell $(i,j)$, the velocity in $x^{(1)}$-direction is reflected.  If $v_{ij,new}^{(1),s}< 0$, the velocity for this particular cell changes its sign. If the sign of the velocity changes in a sufficiently large number of cells, the new mean velocity $\bar{v}^{(1),s}$ is smaller than zero and the entire swarm will
	have negative mean velocity $\overline{v}^{(1),s}$ in $x^{(1)}$-direction and move away from the boundary in the next time step.
	For each experiment, we choose the end of the time horizon $T$ such that the change in the reflection angle after the wall collision is small, i.e., we choose $N_t$ such that $\theta^{r}(t^{Nt}) - \theta^{r}(t^{Nt-1}) < 10^{-3}$.


	\bibliographystyle{aims}
	\bibliography{Literature}
\end{document}